\documentclass{LMCS}

\def\doi{8 (1:01) 2012}
\lmcsheading%
{\doi}
{1--32}
{}
{}
{Jun.~11, 2010}
{Jan.~19, 2012}
{}

\usepackage{algorithm}
\usepackage {latexsym}
\usepackage{listings}
\usepackage{graphicx,color}
\usepackage{url,enumerate,hyperref}

\newcommand{\rr}{\mathbb R}
\newcommand{\br}{\overline{\mathbb R}}
\newcommand{\pp}{\mathcal P}
\newcommand{\parvec}{(\pp(\rd))^n}
\newcommand{\rd}{\rr^d}

\newcommand{\x}{\underline{x}}
\newcommand{\y}{\underline{y}}
\newcommand{\vv}{\underline{v}}
\newcommand{\lya}{\underline{L}}

\newcommand{\FS}{\mathcal{FS}(P,\br)^n}
\newcommand{\fr}{\mathcal{F}(P,\br)}
\newcommand{\fri}{\mathcal{F}(P,\rr\cup\{+\infty\})}
\newcommand{\frf}{\mathcal{F}(P,\rr)}
\newcommand{\frp}{\mathcal{F}(P,\rr_+)}
\newcommand{\ve}{\operatorname{Vex_P}}
\newcommand{\vve}{\operatorname{vex_P}}
\newcommand{\vep}{\ve(P\mapsto\br)}
\newcommand{\ved}{\ve(\rd)}

\def\norm#1{\mbox{$\| #1 \|$}}
\def\clo#1{\vve(#1)}

\def\lfp#1{\operatorname{lfp} #1}

\def\rel#1{#1^{\mathcal R}}
\def\sha#1{#1^{\sharp}}
\def\calR{\mathcal R}
\def\Min{\operatorname*{Min}}
\def\affect{\mathbb{A}}
\def\inter{\mathbb{I}}
\def\union{\mathbb{U}}
\theoremstyle{remark}
\newtheorem{example}[thm]{Example}
\newtheorem{remark}[thm]{Remark}

\begin{document}

\title[Accurate numerical invariants]{\phantom{Coupling policy iteration with
  semi-definite}\vspace{-12 pt}Coupling policy iteration with
  semi-definite relaxation to compute accurate numerical invariants in
  static analysis}

\author[A.~Adj\'e]{Assal\'e Adj\'e\rsuper a}
\address{{\lsuper a}LSV, CNRS \& ENS de Cachan,
61, avenue du Pr{\'e}sident Wilson,
F-94235 Cachan Cedex, France}
\email{assale.adje@lsv.ens-cachan.fr}

\author[S.~Gaubert]{St\'ephane Gaubert\rsuper b}
\address{{\lsuper b}INRIA Saclay and CMAP, Ecole Polytechnique, 
F-91128 Palaiseau Cedex, France}
\email{Stephane.gaubert@inria.fr}

\author[E.~Goubault]{Eric Goubault\rsuper c}
\address{{\lsuper c}CEA, LIST (MeASI), 
F-91191 Gif-sur-Yvette Cedex, France}
\email{Eric.Goubault@cea.fr}
\thanks{{\lsuper{a,b,c}}This work was performed when the first author was with the MeASI team of CEA, LIST and with CMAP, \'Ecole Polytechnique, being supported by a PhD Fellowship of the R\'egion \^Ile-de-France. This work was also partially
supported by the Arpege programme of the French National Agency of Research
(ANR), project ``ASOPT'', number ANR-08-SEGI-005 and by the Digiteo project
DIM08 ``PASO'' number 3389.}

\begin{abstract}
We introduce a new domain for finding precise numerical invariants of
programs by abstract interpretation. This domain, which consists of
sub-level sets of non-linear functions, generalizes the domain of
linear templates introduced by Manna, Sankaranarayanan, and Sipma.  In
the case of quadratic templates, we use Shor's semi-definite
relaxation to derive safe and computable abstractions of semantic
functionals, and we show that the abstract fixpoint can be
over-approximated by coupling policy iteration and semi-definite
programming. We demonstrate the interest of our approach on a series
of examples (filters, integration schemes) including a degenerate one
(symplectic scheme).
\end{abstract}

\maketitle

\section{Introduction}

\label{intro}

We introduce a complete lattice consisting of sub-level sets of
(possibly non-convex) functions, which we use as an abstract
domain in the sense of abstract interpretation \cite{CC77}
for computing
numerical program invariants.
This abstract domain is parametrized by a basis of functions, akin to
the approach put forward by Manna, Sankaranarayanan, and Sipma (the
linear template abstract domain \cite{Sriram1} see also~\cite{Sriram2}), 
except that the basis functions or
templates which we use here need not be linear. The domains obtained in this way encompass the classical abstract
domains of intervals, octagons and linear templates. 

To illustrate the interest of this generalization, 
let us consider a harmonic oscillator: $\ddot{x}+c\dot{x}+x=0$.
By taking an explicit Euler scheme, and for $c=1$
we get the program shown at the left of Figure~\ref{running}.

\begin{figure}
\begin{minipage}{5cm}
\begin{center}
\begin{lstlisting}[frame=single]
x = [0,1];
v: = [0,1];
h = 0.01;
while (true) { [2]
  w = v;
  v = v*(1-h)-h*x;
  x = x+h*w; [3] }
\end{lstlisting}
\end{center}
\end{minipage}
\begin{minipage}{7cm}
\begin{center}
\vskip-3mm\resizebox{5cm}{!}{\input Iteration5bistex}
\end{center}
\end{minipage}
\caption{Euler integration scheme of a harmonic oscillator and the loop invariant found at control point 2}
\label{running}
\end{figure}

The invariant found with our method is shown right of Figure \ref{running}. 
For this, we have considered the template based on functions
$\{x^2,v^2,2x^2+3v^2+2xv\}$, i.e. we consider a domain
where we are looking for upper bounds of these quantities. 
This means that we consider
the {\em non-linear} quadratic homogeneous templates based on $\{x^2,v^2\}$, i.e. symmetric intervals
for each variable of the program, together with the {\em non-linear}
template $2x^2+3v^2+2xv$. The last template comes from the Lyapunov
function that the designer of the {\em algorithm} may have considered
to prove the stability of his scheme, {\em before it has been implemented}. 
This allows us to represent set defined by constraints of the form $x^2\leq c_1$, $v^2\leq c_2$ 
and $2x^2+3v^2+2xv\leq c_3$ where $c_1$, $c_2$ and $c_3$ are degrees of freedom.  
In view of {\em proving the implementation correct}, one is naturally
led to considering such templates\footnote{Of course, as for the linear templates
of~\cite{Sriram1,Sriram2}, we can be interested in automatically finding
or refining
the set of templates considered to achieve a good precision of the abstract
analysis, but this is outside the scope of this article.}. 
Last but not least, it is to be noted
that the loop invariant using intervals, zones, octagons or even polyhedra 
(hence with any linear template)
is the insufficiently strong invariant $h=0.01$
(the variables $v$ and $x$ cannot
be bounded.) 
However, the main interest of the present method
is to carry over to the non-linear setting. For instance, we include in our
benchmarks a computation of invariants (of the same quality) for an implementation
of a highly degenerate example (a symplectic integration scheme, for which
alternative methods fail due to the absence of stability margin).

\paragraph{Contributions of the paper}

We describe the lattice theoretical operations in terms of Galois
connections and generalized convexity in Section~\ref{levsetabs}.
We also show that in the case of a basis of quadratic functions, a good
over-approximation $\rel{F}$ of the abstraction $\sha{F}$ of a semantic functional $F$ can be computed 
by solving a semi-definite program (Section~\ref{quadratic}). 
This over-approximation is obtained using Shor's relaxation.
The advantage of the latter is that it can
be solved in polynomial time to an arbitrary prescribed precision by
the ellipsoid method~\cite{GroetschelLovaszSchrijver1988},
or by interior point methods~\cite{NestNemi} if a strictly feasible
solution is available (however, we warn the reader that interior points
methods, which are more efficient in practice, are known to be polynomial only in the real number model, not in the bit model, 
see the survey~\cite{RaP:97} for more information).

Moreover, we show in Subsection~\ref{politer} that 
the vectors of Lagrange multipliers produced by this relaxation correspond to 
policies, in a policy iteration technique for finding fixpoints or at least
postfixpoints of $\rel{F}$, precisely over-approximating the fixpoints of $\sha{F}$.
Finally, we illustrate on examples (linear recursive filters, numerical integration
schemes) that policy iteration on such quadratic templates is 
efficient and precise in practice, 
compared with Kleene iteration with widenings/narrowings. The fact that
quadratic templates are efficient on such algorithms is generally due to 
the existence of (quadratic) Lyapunov functions that prove their stability.
The method has been implemented as a set of \texttt{Matlab} programs. 

\paragraph{Related work}
This work is to be considered as a generalization of 
\cite{Sriram1}, \cite{Sriram2} 
because it extends the notion of templates to non-linear functions, 
and of 
\cite{Policy1}, \cite{ESOP07}, \cite{Policy2}, \cite{Seidl1} and \cite{Seidl2} since it also
generalizes the
use of policy iteration for better and faster resolution of abstract
semantic equations. Polynomial inequalities (of bounded degree) were
used in \cite{BagnaraR-CZ05} in the abstract interpretation framework 
but the method relies on a reduction to linear inequalities (the polyhedra
domain), hence is more abstract than our domain. 
Particular quadratic inequalities (involving two variables - i.e.
ellipsoidal methods) were used for order 2 linear recursive filters 
invariant generation in \cite{Feret}\footnote{A generalization to order $n$ linear recursive
filters is also sketched in this article.}. 
Polynomial {\em equalities} (and not general functional inequalities as we consider
here) were considered in \cite{Muller,Kapur}.
The use of optimization techniques and relaxations for 
program analysis has also been proposed in
\cite{CousotLagrange}, mostly for synthesising variants for proving
termination, but invariant synthesis was also briefly sketched, with
different methods than ours (concerning the abstract semantics and the
fixpoint algorithm). 
Finally, the interest of using quadratic invariants and in particular
Lyapunov functions for proving control programs correct (mostly in the
context of Hoare-like program proofs) has also been advocated 
recently by E. F\'eron et al. in  \cite{Feron1,Feron2}. 

Finally, we note that a preliminary account of the present work appeared in the conference paper~\cite{adjegaubertgoubault10}.
\section{Lattices of sub-level sets and $P$-support functions }
\label{levsetabs}

We introduce a new abstract domain, parametrized by a basis of
functions ($P$ below). 
The idea is that an abstract value will give the bounds
for each of these functions, hence the name of sub-level sets, with
some abstract convexity condition, Definition \ref{abstractconvexity}. The abstract values 
are computed as the supremum of each function of the basis on a certain sub-level set. The name support functions is 
due to the similarity with the classical support function from convex analysis where the functions 
of the basis are all linear. The idea of using the classical support function in system verification appeared independently in~\cite{AGirard}.

\subsection{$P$-sub-level sets, and their Galois connection with $\pp{(\rr^d)}$}
Let $P$ denote a set of functions from $\rd$ to $\rr$, which is going to be
the basis of our templates. The set $P$ is \emph{not} necessarily \emph{finite} and the functions $p\in P$ 
are \emph{not} necessarily \emph{linear}. We denote by $\fr$
the set of functions $v$ from $P$ to $\br=\rr\cup\{\pm\infty\}$. We equip $\fr$ with the classical partial order for functions i.e.
$v\leq w\iff v(p)\leq w(p)$ for all $p\in P$. We order the set of all subsets of $\rd$ by the subset relation $\subseteq$. 
We define a Galois connection (Proposition~\ref{Galois}) between $\fr$ and the set
of subsets of $\rd$ (made of a concretisation operator $v \mapsto
v^{\star}$, Definition \ref{concretisation} 
and an abstraction operator $C \mapsto C^\dagger$, Definition \ref{abstraction}). This will give the formal background for constructing abstract
semantics using $P$-sub-level sets using abstract interpretation \cite{CC77}, 
in Section~\ref{quadratic}.

\begin{defi}[$P$-sub-level sets]
\label{concretisation}
To a function $v\in\fr$, we associate the $P$-sub-level set denoted 
by $v^{\star}$ and defined as:
\[v^{\star}=\{x\in \rd\mid p(x)\leq v(p),\ \forall p\in P \}\]
\end{defi}

The notion of sub-level set is well known in convex analysis. When $P$ is a set of function that are convex,
the $P$-sub-level sets are convex. In our case, $P$ can contain non-convex functions so $P$-sub-level sets are 
not necessarily convex in the classical sense.

\begin{example}
We come back to the first example and we are focusing on its representation in terms
of $P$-sub-level set.
Let us write, for $(x,v)\in\mathbb{R}^2$, $p_1:(x,v)\mapsto x^2$, $p_2:(x,v)\mapsto v^2$ and
$p_3:(x,v)\mapsto 2x^2+3v^2+2xv$. Let us take $P=\{p_1,p_2,p_3\}$, 
$v(p_1)=3.5000,\ v(p_2)=2.3333$ and $v(p_3)=7$. 
The set $v^{\star}$ is precisely the one shown right of Figure~\ref{running}. 
\end{example}

\begin{example}
We next show some $P$-sub-level sets which are not convex in the
usual sense.      
Let us write, for $(x,y)\in\mathbb{R}^2$, $p_1:(x,y)\mapsto -y^2-(x+2)^2$, $p_2:(x,y)\mapsto -y^2-(x-2)^2$ and
$p_3:(x,y)\mapsto -(y-2)^2-x^2$,  $p_4:(x,y)\mapsto -(y+2)^2-x^2$. Let us take $P=\{p_1,p_2,p_3,p_4\}$ and 
$v(p_1)=v(p_2)=v(p_3)=v(p_4)=-2$. The set $v^{\star}$
is shown Figure~\ref{ex2}.

\begin{figure}
\begin{minipage}{6cm}
\begin{center}
\resizebox{6cm}{!}{\input exqua2tex}
\caption{A $P$-sub-level set arising from non-convex quadratic functions.}
\label{ex2}
\end{center}
\end{minipage}
\begin{minipage}{6cm}
\begin{center}
\resizebox{6cm}{!}{\input exabs0tex}
\caption{A $P$-sub-level set arising from linear forms.}
\label{ex5}
\end{center}
\end{minipage}
\end{figure}
\end{example}

In our case, $P$ is a set of functions from $\rd$ to $\rr$ possibly non-linear, so we generalize the 
concept of support functions (e.g. see Section 13 of~\cite{Roc}). 
The idea of generalizing support functions 
to the non-linear case is not new and this extension is due to Moreau \cite{Moreau}.


\begin{defi}[$P$-support functions]
\label{abstraction}
To $X\subseteq\rd$, we associate the $P$-support function denoted by
$X^{\dag}$ and defined as:
\[X^{\dag}(p)=\sup_{x\in X}p(x)\]
\end{defi}

\begin{prop}
\label{Galois}
The pair of maps $v\mapsto v^{\star}$ and $X\mapsto X^{\dag}$ defines a 
Galois connection between $\fr$ and the set of subsets of $\rd$.
\end{prop}

\begin{proof}
First, the functions $v\mapsto v^{\star}$ and $X\mapsto X^{\dag}$ are clearly monotonic. Now, we have to show that, 
for all $X\subseteq\rd$ and all $v\in\fr$:
$X\subseteq v^{\star}\iff X^{\dag}(p)\leq v(p),\ \forall\, p\in P$.
Let $X\subseteq\rd$, $v\in\fr$ and $p\in P$. We assume that $X\subseteq v^{\star}$. This implies that 
$X^{\dag}(p)=\sup\{p(x)\mid x\in X\}\leq \sup\{p(x)\mid q(x)\leq v(q),\ \forall q\in P\}$. But,
for every $x\in v^{\star}$, $p(x)\leq v(p)$, so $\sup\{p(x)\mid q(x)\leq v(q),\ \forall q\in P\}\leq v(p)$,
hence, $X^{\dag}(p)\leq v(p)$ and the first implication is shown.
Now, we assume that $X^{\dag}(p)\leq v(p),\ \forall\, p\in P$. Let $x\in X$, we have, for all $p\in P$,
$p(x)\leq X^{\dag}(p)$, thus $p(x)\leq v(p)$ and finally $x\in v^{\star}$.
\end{proof}

\begin{remark}
\label{galoisprop}
In the proof, we showed $(v^\star)^\dag\leq v$, we can similarly prove that $X\subseteq (X^\dag)^\star$.
Using this remark and using the monotonicity of $v\mapsto v^{\star}$ and $X\mapsto X^{\dag}$, we get
$v^\star\subseteq ((v^\star)^\dag)^\star\subseteq v^\star$ and similarly, $((X^\dag)^\star)^\dag=X^\dag$.
These properties are well-known in Galois connection theory.
\end{remark}

In the terminology of abstract interpretation, $(.)^\dag$ is the abstraction function, usually denoted by $\alpha$,
and $(.)^\star$ is the concretisation function, usually denoted by $\gamma$.

\subsection{The lattices of $P$-convex sets and $P$-convex functions}

The sets of points in $\rr^d$ which are exactly represented by
their corresponding $P$-sub-level sets are called $P$-convex sets, as in the definition below.   
These can be identified with the set of {\em abstract elements}
we are considering\footnote{Formally, this is the upper-closure in $\pp(\rr^d)$ of the
set of abstract elements.}. We show in Theorem~\ref{lattice} that they
constitute a complete lattice. We use the terminology $P$-convex because of the analogy with the abstract convexity defined 
by Moreau \cite{Moreau}, Singer \cite{Singer} and Rubinov \cite{Rubinov}. These authors define 
convexity without linearity from a given family of (non-linear) functions. Moreau called this notion of convexity,
$B$-convexity where $B$ is the family of functions whereas Singer and Rubinov use the term abstract convexity. 
The $P$-convexity defined in this paper corresponds to the classical notions of closure from the Galois connections theory 
and also to the abstract convexity introduced by Moreau.  


\begin{defi}[$P$-convex hull]
The $P$-convex hull of an element $v\in\fr$ is the function $\clo{v}$ which is equal to $(v^\star)^{\dag}$. 
Similarly, the $P$-convex hull of a subset $X$ is the set $\clo{X}$ which is equal to $(X^\dag)^{\star}$.
\end{defi}

\begin{example}
\label{exampletri}
Let us consider the triangle, depicted in Figure~\ref{ex5}. Let us take $P_1=\{(x,y)\mapsto y+x,\ (x,y)
\mapsto x-y,\ (x,y)\mapsto -x\}$. Its $P_1$-convex hull is the one
depicted Figure~\ref{ex6}. If we take instead $P_2=\{(x,y)\mapsto y^2-x^2,\ (x,y)
\mapsto x,\ (x,y):\mapsto -x\}$, its $P_2$-convex hull is shown in Figure~\ref{ex7}.

\begin{figure}
\begin{minipage}{6cm}
\begin{center}
\resizebox{2cm}{!}{\input hulltrian1tex} 
\caption{$P_1$-convex hull of example \ref{exampletri}}
\label{ex6}
\end{center}
\end{minipage}
\begin{minipage}{6cm}
\begin{center}
\resizebox{2cm}{!}{\input hulltrian2tex} 
\end{center}
\caption{$P_2$-convex hull of example \ref{exampletri}.}
\label{ex7}
\end{minipage}
\end{figure}
\end{example}

\begin{defi}[$P$-convexity]
\label{abstractconvexity}
Let $v\in\fr$, we say that $v$ is a $P$-convex function if $v=\clo{v}$.
A set $X\subseteq\rd$ is a $P$-convex set if $X=\clo{X}$.
\end{defi}

\begin{example}
\label{trian}
Let us come back to the triangle depicted in Figure~\ref{ex5}. If $P$ is the set of
linear forms defined by the 
faces of this triangle
i.e. $P$ consists of the maps $(x,y):\mapsto y-x$, $(x,y):\mapsto y+x$ and $(x,y):\mapsto -y$,
then it is a $P$-convex set. 
But if $P$ is, for example, linear forms defined by orthogonal
vectors to the faces of the triangle, the previous triangle is no longer an $P$-convex set.  
\end{example}

Abstract (P-) convexity is a special instance of Galois connection. The latter
are classically used in abstract interpretation. In particular, the theory
of Galois connections yields the following result:
the $P$-convex hull of a function $v\in\fr$ is the greatest $P$-convex function which is smaller than $v$ and dually
the $P$-convex hull of a subset $X$ is the smallest $P$-convex set which is greater than $X$.  

We respectively denote by $\vep$ and $\ved$ the set of all $P$-convex function of $\fr$ and
the set of all $P$-convex sets of $\rd$.
\begin{defi}[The meet and join]
Let $v$ and $w$ be in $\fr$. We denote by $\inf(v,w)$
and $\sup(v,w)$ the functions defined respectively by:
\[ 
p\mapsto\inf(v(p),w(p))\text{ and } p\mapsto\sup(v(p),w(p))\enspace .
\]
We equip $\vep$ with the meet (respectively join) operator:
\begin{equation}
v\vee w=\sup(v,w)
\label{equnion}
\end{equation}
\begin{equation}
v\wedge w =(\inf(v,w)^{\star})^{\dag}
\label{eqinter}
\end{equation}
Similarly, we equip $\ved$ with the two following operators: 
$X\sqcup Y=((X\cup Y)^{\dag})^{\star}$, 
$X\sqcap Y =X\cap Y$.
\end{defi}

The family of functions $\vep$ is ordered by the partial order of real-valued
functions i.e. $v\leq w\iff v(p)\leq w(p)\ \forall p\in P$.
The family of set $\ved$ is ordered by the subset relation denoted by $\subseteq$.
The next theorem follows readily from the fact that the pair of functions $v\mapsto v^{\star}$ and $X\mapsto X^{\dag}$ defines a 
Galois connection, see e.g.~\cite[\S~7.27]{priestley}.
\begin{thm}
\label{lattice}
$(\vep,\wedge,\vee)$ and $(\ved,\sqcap,\sqcup)$ are isomorphic complete lattices. \hfill\qed
\end{thm}

\subsection{Intervals, Zones, Octagons and Sankaranarayanan et al.'s linear templates}
The domain of intervals arises as a special domain of $P$-convex sets, in which the basis is $P=\{x_1,-x_1,\ldots,
x_n,-x_n\}$ where $x_i$ ($i=1,\ldots,n$) are the program variables. An abstract value $v$ in our domain encodes the supremum $v(x_i)$ and the infimum
$-v(-x_i)$ of an interval for the variable $x_i$.

Zones and octagons are treated in a similar manner. For instance, for zones, take $P=
\{x_i - x_j \mid i, j=0,\ldots,n, i\neq j\}$, adding a dummy variable $x_0$ (always equal to 0), as
customary, to subsume intervals. 
Of course, linear templates as defined in \cite{Sriram1} are particular $P$-convex sets, for which
$P$ is given by a finite set of linear functionals.

We remark that in the case of zones,
$v(x_i-x_j)$ is exactly the entry $i,j$ of the DBM
(Difference Bound Matrix) representing the corresponding zone. 
Also, elements of $\vep$ corresponding naturally to {\em closed} DBMs, that is, canonical forms of
DBMs. As is well known \cite{PhDMine}, the union of two zones preserves closure whereas the intersection
does not necessarily preserve closure. More generally, the set
of closed elements for a Galois connection is stable under the join
operation. This is reflected in our domain by~\eqref{equnion} 
and~\eqref{eqinter}.

\section{Quadratic templates}
\label{quadratic}
In this section, we instantiate the set $P$ to  
linear and quadratic functions. This allows us to give a systematic way
to derive the abstract semantics of a program. The main result is
that the abstract semantics for an assignment and for a test, can be safely over-approximated  
by Shor's relaxation scheme, Theorem~\ref{polytime}.
\begin{defi}
\label{quadfct}
We say that $P$ is a quadratic zone iff every element template $p\in P$ can be written as:
\[
x\mapsto p(x)=x^{T}A_p x+b_p^{T} x,
\]
where $A_p$ is a $d\times d$ symmetric matrix (in particular $A_p$ can be a zero matrix),
$x^{T}$ denotes the transpose of a vector $x$, $b_p$ is a $\rd$ vector.
\end{defi}

Now, we suppose that $P$ is \emph{finite}. We denote by $\frf$ the set of functions from $P$ to $\rr$, 
$\fri$ the set of functions from $P$ to $\rr\cup\{+\infty\}$ and $\frp$ the set of functions from $P$ 
to $\rr_+$.

\vskip .5cm

Suppose now we are given a program with $d$ variables $(x_1,\ldots,x_d)$ 
and $n$ control points numbered from $1$ to $n$.
We suppose this program is written in a simple toy version of a \texttt{C}-like imperative
language, comprising global variables, no procedures, assignments of variables using only
{\em parallel assignments}
$(x_1,\ldots,x_d)=T(x_1,\ldots,x_d)$, tests of the form
$(x_1,\ldots,x_d) \in C$, where $C$ is some shape in $\pp(\rr^d)$, and while loops 
with similar entry tests. We do not recapitulate the standard collecting semantics that
associates to this program a  
monotone map $F: \parvec \rightarrow \parvec$ 
whose least fixpoints $\lfp(F)$ has as $i$th component ($i=1,\ldots,n$) the
subset of $\rr^d$ of values that the $d$ variables $x_1, \ldots, x_d$ can take at control
point $i$. 

The aim of this section is to compute, inductively on the syntax, the
abstraction (or a good over-approximation of it) $\sha{F}$
of $F$ from $\fr^n$ to itself defined as usual as:
\begin{equation}
\label{absmap}
\sha{F}(v)=(F(v^\star)^\dagger) 
\end{equation}

The notation $v^\star$ is in fact the vector of sets $(v_1^\star,\cdots,v_n^\star)$ and 
$(F(v^\star)^\dagger)$ is also interpreted component-wise. The notation $\clo{v}$ will 
be also understood component-wise.

\subsection{Shor's semi-definite relaxation scheme}
\label{shorscheme}

Finding the maximal value of a non-concave quadratic function under convex or non-convex 
quadratic constraints is known to be an \emph{NP-Hard} problem,
see~\cite{Vavasis90} for a discussion of complexity issues in quadratic programming.     
Shor's relaxation scheme (see ~\cite[Section 4.3.1]{BTN} or Shor's original 
article~\cite{shor} for details) consists of over-approximating the value of a general
quadratic optimization problem by the optimal value of a semi-definite programming (SDP for short) problem, the latter being computationally more tractable.

Indeed, SDP problems
can be solved in polynomial time to an arbitrary prescribed precision by
the ellipsoid method~\cite{GroetschelLovaszSchrijver1988}.
More precisely, let $\epsilon>0$ be a given rational, suppose that the 
input data of a semi-definite program are rational and suppose that an integer
$N$ is known, such that the feasible set lies inside the ball of the radius $N$ around zero.
Then an $\epsilon$-optimal solution 
(i.e., a feasible solution the value of which is at most at a distance $\epsilon$ from the optimal value) can be found in a time that is polynomial in the number of bits of the input data and in $-\log\epsilon$. Moreover, 
an $\epsilon$-solution of an SDP problem can also be found in polynomial time
by interior point methods~\cite{NestNemi} if a strictly feasible solution is available.  However, when the input is rational, no size on the bit lengths of the intermediate data is currently known, so that the term ``polynomial time'' 
for interior point methods is only understood in the model of computation over real numbers (rather
than in bit model~\cite{GareyJohnson}). The advantage of interior methods is that they are very efficient in practice. We refer the reader to~\cite{RaP:97} for more information.



Let $f$, $\{f_i\}_{i=1,\ldots,m}$ be quadratic functions on $\rd$.
Let us consider the following constrained maximization problem:
\begin{equation}
\label{quadpb}
\sup\{f(x)\mid f_i(x)\leq 0,\,\forall i=1,\ldots,m\}
\end{equation}

In constrained optimization, it is classical to construct another constrained optimization problem 
from the initial one in order to solve an easier problem. A technique called Lagrange duality (for details see for example 
~\cite[Section 5.3]{AuTe}) consists in adding to the objective function the inner product of the vector of constraints with 
a positive vector of the euclidean space whose the dimension is the number of constraints. 
In our context, the value of \eqref{quadpb} is given by the following sup-inf (primal) value \eqref{supinfshor}: 
\begin{equation}
\label{supinfshor}
\sup_{x\in\rd}\inf_{\lambda\in\rr_+^m} f(x)-\sum_{i=1}^m \lambda_i f_i(x)\enspace .
\end{equation}
 
A simple result of constrained optimization called weak duality theorem ensures that if we commute the $\inf$ and 
the $\sup$ in the formula \eqref{supinfshor}, the result is greater than the value \eqref{supinfshor}.
The commutation of the $\inf$ and the $\sup$ gives us the so called (dual) value: 
\begin{equation}
\label{infsupshor}
\inf_{\lambda\in\rr_+^m}\sup_{x\in\rd} f(x)-\sum_{i=1}^m \lambda_i f_i(x)\enspace .
\end{equation}

The vectors $\lambda\in\rr_+^m$ are called vectors of Lagrange multipliers. The function 
$\lambda\mapsto \sup_{x\in\rd} f(x)-\sum_{i=1}^m \lambda_i f_i(x)$ is always convex and lower semi-continuous 
(these definitions are recalled in the appendix), so it has good properties to minimize it. If the function $f$ 
is concave, the functions $f_i$ are convex and if the Slater constraint qualification (i.e. there exists 
$x\in\rd$ such that $f_i(x)<0$ for all $i=1,\ldots,m$) holds then \eqref{supinfshor} and \eqref{infsupshor} 
coincide (this will be used  Proposition~\ref{sdualth}).

Shor's relaxation scheme consists in computing the value \eqref{infsupshor} by solving a semi-definite program.
We introduce the matrix $M(g)$, for a quadratic function $g$ written as $x^TA_gx+b_g^Tx+c_g$ and the matrix $N(y)$ for a real 
$y$ defined as:
\begin{equation}
\label{shmatrix}
M(g)=\begin{pmatrix} c_g\, & \frac{1}{2}b_g^T \\ \frac{1}{2}b_g\, & A_g \end{pmatrix},\ \mathrm{and}\  
N_{1,1}(y)=y,\ N_{i,j}(y)=0\ \mathrm{if}\ (i,j)\neq (1,1)
\end{equation}

Let $\preceq$ denote the L\"owner ordering of symmetric matrices, so that $A\preceq B$ iff
all eigenvalues of $B-A$ are non-negative. 

When we fix $\lambda\in\rr_+^m$, we have to maximize an unconstrained quadratic problem
and the the maximum is finite iff there exists $\eta\in\rr$ such that, for all $x\in\rd$,
$f(x)-\sum_{i=1}^m \lambda_i f_i(x)\leq \eta$ and since $f$, $f_i$ are quadratic functions,
this is equivalent to 
$x^T(A_f-\sum_{i=1}^m \lambda_i A_{f_i})x+(b_f-\sum_{i=1}^m \lambda_i b_{f_i}+
c_{f}-\sum_{i=1}^m \lambda_i c_{f_i}-\eta\leq 0$ for all $x\in\rd$ which is equivalent to the 
fact that the matrix $M(f)+\eta N(-1)-\sum_{i=1}^m \lambda_i M(f_i)$ is negative semi-definite.
Consequently, taking the infimum over $\lambda\in\rr_+^m$, we recover the value \eqref{infsupshor}.  

In conclusion, Shor's relaxation scheme consists in solving the following SDP problem:
\begin{equation}
\label{shordual}
\Min_{\substack{\lambda\in\rr_+^m\\ \ \eta\in\rr}} \eta\mbox{ s.t. }\, M(f)+\eta N(-1)-\sum_{i=1}^m \lambda_i M(f_i)]
\preceq 0
\end{equation}
which is equal to the value \eqref{infsupshor}, hence an over-approximation of the optimal value
of the problem~\eqref{quadpb}. We can use a verified SDP solver as VSDP~\cite{jan-cha-kei-07} to solve a SDP problem.

\begin{example}
\label{rotation}
We consider the rotation of angle $\phi\in (0,\pi)$ leaving invariant the unit circle $S^1=\{(x,y)\in \rr^2\mid x^2+y^2=1\}$.
We write $T$ the linear map associated to the rotation matrix $A$: 
\[
T
\begin{pmatrix}
 x\\
 y
\end{pmatrix}
=A (x,y)^T
=
\begin{pmatrix}
\cos\phi& &-\sin\phi\\
\sin\phi& &\cos\phi
\end{pmatrix}
\begin{pmatrix}
x\\
y
\end{pmatrix}
\]
where $x^2+y^2=1$. To show that the proposed relaxation is accurate we next observe that it preserves the fact that the 
unit sphere is invariant by rotation i.e. $T(S^1)=S^1$. The inclusions $S^1\subseteq T(S^1)$ and  $T(S^1)\subseteq S^1$ 
can be proved by the same manner because $A$ is an orthogonal matrix i.e. 
$A^T A=A A^T=Id$: $((x,y)\in S^1 \implies (x,y)\in T(S^1))$ $\iff$ ($A^T(x,y)\in S^1$ for $(x,y)\in S^1$), 
hence, to show $S^1\subseteq T(S^1)$ is reduced to show $T^*(S^1)\subseteq S^1$ where $T^*$ is the linear map 
associated to $A^T$. We only prove that $T(S^1)\subseteq S^1$ and we use the Shor's relaxation scheme to prove that.

We introduce the set of quadratic functions $P=\{p_1(x,y)\mapsto x^2+y^2,\; p_2(x,y)\mapsto -(x^2+y^2)\}$ 
and we set $v_1(p_1)=1$ and $v_1(p_2)=-1$. The unit sphere is the $P$-sub-level set of $v_1$: 
$\{(x,y)\in \rr^2\mid p(x,y)\leq v_1(p),\ \forall p\in P\}$. To show $T(S^1)\subseteq S^1$, it suffices to show 
$T(v_1^\star)^\dag=v_1$. We write:
\begin{eqnarray*}
v_2(p_1):=T(v_1^\star)^\dagger(p_1)=\sup\{p_1(T(x,y))\mid p_1(x,y)\leq 1,\; p_2(x,y)\leq -1\}\\
v_2(p_2):=T(v_1^\star)^\dagger(p_2)=\sup\{p_2(T(x,y))\mid p_1(x,y)\leq 1,\; p_2(x,y)\leq -1\}
\end{eqnarray*}
For the calculus of $v_2(p_1)$, the matrices defined in \ref{shmatrix} are:
\[
M(p_1\circ T)=\begin{pmatrix} 0 & 0 & 0 \\ 0 & 1 & 0\\ 0 & 0 & 1\end{pmatrix},
\ M(p_2\circ T)=-M(p_1\circ T),
\]
\[
M(p_1)=\begin{pmatrix} 0 & 0 & 0 \\ 0 & 1 & 0\\ 0 & 0 & 1\end{pmatrix}\text{ and }
M(p_2)=-M(p_1)\enspace .
\]
By using the Shor's relaxation scheme, the SDP problem~\eqref{shordual} gives the following equalities:
\[
v_2(p_1)=
\Min_{\substack{\lambda(p_1)\geq 0\\ \lambda(p_2)\geq 0\\ \eta\in\rr}} \eta\mbox{ s.t. }\, 
M(p_1\circ T)+\eta N(-1)+\lambda(p_1) (N(1)-M(p_1))+\lambda(p_2) (N(-1)-M(p_2))\preceq 0
\]
and
\[
v_2(p_2)=
\Min_{\substack{\lambda(p_1)\geq 0\\ \lambda(p_2)\geq 0\\ \eta\in\rr}} \eta\mbox{ s.t. }\, 
M(p_2\circ T)+\eta N(-1)+\lambda(p_1) (N(1)-M(p_1))+\lambda(p_2) (N(-1)-M(p_2))\preceq 0
\]
We can rewrite the two previous SDP problems as follows:
\[
\noindent v_2(p_1)=\Min_{\substack{\lambda(p_1)\geq 0\\ \lambda(p_2)\geq 0\\ \eta\in\rr}}\eta\ \mathrm{s.t.}
\begin{pmatrix}-\eta+\lambda(p_1)-\lambda(p_2)&0&0\\0&1-\lambda(p_1)+\lambda(p_2)&0\\
0&0&1-\lambda(p_1)+\lambda(p_2)\end{pmatrix}\preceq 0
\]                                                              
and
\[
v_2(p_2)=\Min_{\substack{\lambda(p_1)\geq 0\\ \lambda(p_2)\geq 0\\ \eta\in\rr}}\eta\ \mathrm{s.t.}
\begin{pmatrix}-\eta+\lambda(p_1)-\lambda(p_2)&0&0\\0&-1-\lambda(p_1)+\lambda(p_2)&0\\
0&0&-1-\lambda(p_1)+\lambda(p_2)\end{pmatrix}\preceq 0
\]
To solve these optimization problems, we could call an SDP solver, but in this case,
it suffices to solve a system of inequalities:
all the diagonal elements must be non-positive, for example, for the first problem, this
implies that $\eta\geq 1$ and since we minimize $\eta$ we get $\eta=1$.
Hence, we find $v_2(p_1)=1$ and $v_2(p_2)=-1$ and finally we have proved $T(v_1^\star)^\dag=v_1$ and we conclude that
the rotation invariance of the unit sphere $S^1=\{(x,y)\in \rr^2\mid x^2+y^2=1\}$ is preserved for the relaxation. 
\end{example}

\subsection{Abstraction of assignments and tests using Shor's relaxation}
\label{assigntestsect}
\subsubsection{Abstraction of assignments}
\label{assignsect}
We focus on assignments $(x_1,\ldots,x_d)=T(x_1,\ldots,x_d)$ at control point $i$ such that $p\circ T$ is a quadratic 
function for every $p\in P$. Equation \eqref{absmap} translates in that case to (given that $v_{i-1}$ defines the abstract
value at control point $i-1$, i.e. immediately before the assignment):
\newcommand{\mybrackets}[1]{\big(#1\big)}
\begin{equation}
\label{optpb}
\mybrackets{\sha{F_i}(v)}(p)=\sup\{p\circ T(x)\mid q(x)-v_{i-1}(q)\leq 0, \forall q\in P\} \enspace .
\end{equation}
We recognize the constrained optimization problem~\eqref{quadpb} and we use Lagrange duality as in the first step
of Subsection~\ref{shorscheme}. In our case, the Lagrange
multipliers are some non-negative functions $\lambda$ from $P$ to $\rr$. 
We thus consider the function which we will call the \emph{relaxed} function:
\begin{equation}
\label{eqn}
\mybrackets{\rel{F_i}(v)}(p):=\inf_{\lambda\in \frp}\sup_{x\in\rd} p\circ T(x)+\sum_{q\in P}\lambda(q)\mybrackets{v_{i-1}(q)-q(x)} \enspace .
\end{equation}
To compute $\rel{F_i}$, we apply Shor's relaxation scheme of Subsection~\ref{shorscheme} and particularly the reformulation
as the SDP problem \eqref{shordual}, 
we get:
\begin{equation}
\label{SDP2}
\mybrackets{\rel{F_i}(v)}(p)=\Min_{\substack{\lambda\in\frp\\ \ \eta\in\rr}} 
\eta\, \mbox{ s.t. }\, M(p\circ T)+\eta N(-1)+\sum_{q\in P} \lambda(q) \mybrackets{N(v_{i-1}(q))-M(q)}
\preceq 0  
\end{equation}
where $M(p\circ T)$, $N(-1)$ and $M(q)$ are the matrices defined in~\eqref{shmatrix}.

We can treat the case where the map $T$ has non-linear quadratic components. 
The condition $p\circ T$ is quadratic for all $p\in P$ implies, in this case, that the quadratic 
templates $p$ should be linear forms which is equivalent to the fact that the matrices $M(p)$ 
have the following form:
\begin{equation}
\label{noquadcont}
\begin{pmatrix} 0\, & \frac{1}{2}b_p^T \\ \frac{1}{2}b_p\, & 0 \end{pmatrix}
\end{equation}
and then the matrix $\eta N(-1)+\sum_{q\in P} \lambda(q) \mybrackets{N(v_{i-1}(q))-M(q)}$ in Equation~\eqref{SDP2}
has no "quadratic counterpart" (of the form of Equation~\eqref{noquadcont}, we  conclude that the SDP 
minimization problem is feasible in this case iff $M(p\circ T)$ is negative semi-definite. 
Finally, the problem of Equation~\eqref{optpb} is concave (since the forms $q\in P$ are linear) 
and consequently the Shor's relaxation scheme computes exactly the abstract semantic functional 
of Equation~\eqref{optpb}. 

\begin{figure}[ht!]
\begin{minipage}{9cm}
\begin{center}
\begin{lstlisting}[frame=single]
x  = [0,10];
y  =      1; [1]  
xn = -3*x*x-y*y;
yn = -y*y+x*x; 
x  = xn;
y  = yn [2]}
\end{lstlisting}
\end{center}
\end{minipage}
\caption{A simple program written in polynomial arithmetic}
\label{NonLinearEx}
\end{figure}
\begin{example}[Non linear assignment]
\label{exampleNLprog}
Let us take a simple program written in a polynomial arithmetic
(assignments now involve polynomial expressions in the variables,
rather than linear ones as in the previous example).
This program is described at Figure \ref{NonLinearEx}.
We consider the set of templates $P=\{p_1,\ p_2\}$, with $p_1:(x,y)\mapsto x+y$ and $p_2:(x,y)\mapsto x-y$. 
We define the function $T$ as follows:
\[  
(x,y)\mapsto T(x,y)=
\left(
\begin{array}{c}
-3x^2-y^2\\
-y^2+x^2
\end{array}
\right)
\]
which is the non-linear assignment of the program described at Figure \ref{NonLinearEx}.
By Equation ~\eqref{absmap}, the abstract semantics are, for $w\in\fr^2$: 
\begin{align*}
\mybrackets{\sha{F_1}(w)}(p)
&=(11,9)\\
\mybrackets{\sha{F_2}(w)}(p)
&=\displaystyle{\sup_{(x,y)\in(w_1)^{\star}}} (p\circ T)(x,y)
\end{align*}

About the first component of the abstract semantics, the vector $(11,9)$ means that 
$\mybrackets{\sha{F_1}(w)}(p_1)=11$ and $\mybrackets{\sha{F_1}(w)}(p_2)=9$. This 
corresponds to the relations $x+y\leq 11$ and $x-y\leq 9$ at control point [1].
At control point [2], the relation $xn+yn=-2x^2-2y^2$ holds between the variables at control point [1] and [2]. In this special case, the assignments are quadratic and the templates are linear, so the evaluation of the abstract semantic 
functional still reduces to a (non-convex) quadratic programming problem, which,
by application of Shor relaxation, leads to
the following relaxed functional $\mybrackets{\rel{F_2}(w)}(p_1)$ at control point 2, for an element $w\in\fr^2$ and  for the template $p_1$:
\[
\mybrackets{\rel{F_2}(w)}(p_1)=
\inf_{\substack{\lambda(q)\geq 0\\ \forall\, q\in P\geq 0}}\sup_{(x,y)\in\rr^2}\sum_{q\in P} \lambda(q)w_1(q)
+(x,y)\begin{pmatrix} -2&0\\ 0&-2\end{pmatrix}(x,y)^T-\sum_{q\in P}\lambda(q) q(x,y)
\]
By introducing the matrices
\[
M(p_1)=\begin{pmatrix}
0 & 0.5 & 0.5\\
0.5 & 0 & 0\\
0.5 & 0 & 0
\end{pmatrix},\qquad
M(p_2)=\begin{pmatrix}
0 & 0.5 & -0.5\\
0.5 & 0 & 0\\
-0.5 & 0 & 0
\end{pmatrix},
\]
\[
M(p_1\circ T)=\begin{pmatrix}
0 & 0 & 0\\
0 & -2 & 0\\
0 & 0 & -2
\end{pmatrix}
\]
and using Equation \eqref{SDP2}, we get
\[
\mybrackets{\rel{F_2}(w)}(p_1)=\Min_{\substack{\lambda(p)\geq 0\\ \forall p\in P \\ \eta\in\rr}} 
\eta\, \mbox{ s.t. }\, M(p_1\circ T)+\eta N(-1)+\sum_{q\in P} \lambda(q) \mybrackets{N(w_{1}(q))-M(q)}
\preceq 0   \enspace .
\]
Finally, using \texttt{Matlab}\footnote{Matlab is a registered trademark of the MathWorks,Inc.}, 
\texttt{Yalmip} ~\cite{Yalm} and \texttt{SeDuMi}~\cite{Sedu}, we find 
$\mybrackets{\rel{F_2}(w)}(p_1)=-3.2018e-09\simeq 0$ which means, at control point [2], that $x\leq -y$ and the image of the
second template by the relaxed function is $\mybrackets{\rel{F_2}(w)}(p_2)=-1.0398e-09\simeq 0$ which means, 
at control point [2], $x\leq y$. 
The invariant found $\{(x,y)\mid x\leq -y, x\leq y\}$ is an unbounded set. We can refine this set from 
the invariant found by interval arithmetic to get a bounded set. 
\end{example}



\subsubsection{Abstraction of simple tests}
\label{testsect}
Now, we focus on a simple test and we write $j$ the control point of the test. We assume here that a test is 
written as $r(x_1,\ldots,x_d) \leq 0$  where $r$ is a quadratic function. We assume that the operation for the 
``\texttt{then}'' branch has the form $x=T_{then}(x)$ and the operation for the ``\texttt{else}'' branch has the form 
$x=T_{else}(x)$ where $T_{then}$, $T_{else}$ such that $p\circ T_{then}$ and $p\circ T_{else}$ are quadratic for all $p\in P$.
To enter into the ``\texttt{then}'' branch, the abstract values of control point $j-1$ must satisfy the test condition of the 
``\texttt{then}'' branch, so the abstraction of a test is $F_{j}(X)=T_{then}(X_{j-1})\cap \{x\in \rd\mid r(x)\leq 0\}$ for the 
``\texttt{then}'' branch. 
For the ``\texttt{else}'' branch, we have similarly $F_{j+2}(X)=T_{else}(X_{j-1}) \cap \{x\in\rd\mid r(x)>0\}$ which 
is equivalent to $F_{j+2}(X)=T_{else}(X_{j-1}) \cap \{x\in\rd\mid -r(x)<0\}$. However, we cannot use the Shor's 
relaxation scheme with strict inequalities and so we replace the set $\{x\in\rd\mid -r(x)<0\}$ by the set 
$\{x\in\rd\mid -r(x)\leq 0\}$ which is larger so we compute at least a safe over-approximation. When the function $r$ 
is concave and the set $\{x\in\rd\mid -r(x)<0\}$ is non-empty, the closure of the former set coincides with the latter set. 
For the ``\texttt{else}'' branch, the abstraction is, finally, $F_{j+2}(X)=T_{else}(X_{j-1}) \cap \{x\in\rd\mid -r(x)\leq 0\}$.
As we deal with arbitrary quadratic functions $r$, it is sufficient to show here how to deal with the equations at 
control point $j$ and we simply write $T$ instead of $T_{then}$.  
By using Equation~\eqref{absmap}, we get, for $v\in\fr^n$, and $p\in P$:
\[
\mybrackets{\sha{F_{j}}(v)}(p)=\left(T\left(v_{j-1}^\star \cap \{x\in\rd\mid r(x)\leq 0\}\right)\right)^\dagger(p)
\]
then, by a simple calculus:
\begin{equation}
\label{eqnabstracttest}
\mybrackets{\sha{F_{j}}(v)}(p)=\sup\{p\circ T(x)\mid q(x)\leq v_{j-1}(q)\ \forall\, q\in P,\ r(x)\leq 0\}.
\end{equation}
Using again the Lagrange duality, we get the relaxed problem:
\begin{equation}
\label{eqntest}
\mybrackets{\rel{F_{j}}(v)}(p):=\inf_{\substack{\lambda\in \frp\\ \mu\in \rr_+}}\sup_{x\in\rd}\sum_{q\in P}\lambda(q)v_{j-1}(q)
+ p\circ T(x)-\sum_{q\in P}\lambda(q)q(x)-\mu r(x).
\end{equation}
Using the Shor's relaxation scheme described in Subsection~\ref{shorscheme} and in particular the SDP problem \eqref{shordual},
we can rewrite $\mybrackets{\rel{F_j}(v)}(p)$ as the following SDP problem:
\begin{equation}
\label{SDPinter}
\Min_{\substack{\lambda\in\frp\\ \mu\in\rr_+\\ \eta\in\rr}} 
\eta\, \mbox{ s.t. }\, M(p\circ T)+\eta N(-1)+\sum_{q\in P} \lambda(q) \mybrackets{N(v_{i-1}(q))-M(q)}-\mu M(r)\preceq 0
\end{equation}
\noindent where $M(p\circ T)$, $N(-1)$, $M(q)$ and $M(r)$ are the matrices defined in~\eqref{shmatrix}.
\begin{figure}[ht!]
\begin{minipage}{9cm}
\begin{center}
\begin{lstlisting}[frame=single]
x = [0,10];
y =      1;
u =      0; [1] 
if (y*y+x*x-2<=0){
    u=x;
    x = 3-y*y;
    y = u-1;[2]
    }
else{ 
     x=1-y*y;
     y=3-y;[3]
    }
\end{lstlisting}
\end{center}
\end{minipage}
\caption{A simple program written in polynomial arithmetic}
\label{NonLinearExtest}
\end{figure}
\begin{example}[Quadratic assignment and quadratic tests]
Let us take the simple program written in a polynomial arithmetic described at Figure \ref{NonLinearExtest}. 
We introduce the quadratic function $r:(x,y)\mapsto y^2+x^2-2$ which represents the test. We define the function $T_{then}$ as 
follows:
\[  
(x,y)\mapsto T_{then}(x,y)=
\left(
\begin{array}{c}
3-y^2\\
x-1
\end{array}
\right)\enspace .
\]
We also define the function $T_{else}$ by:
\[  
(x,y)\mapsto T_{else}(x,y)=
\left(
\begin{array}{c}
1-y^2\\
3-y
\end{array}
\right)\enspace .
\]
These two functions represent the assignments of the branches \texttt{then} and \texttt{else} respectively.
Similarly to Example ~\ref{exampleNLprog}, we consider the domain of intervals: the set of linear 
templates $P=\{\x,-\x,\y,-\y\}$ where $\x:(x,y)\mapsto x$, $\y:(x,y)\mapsto y$. 
The abstract semantics defined by Equation \eqref{absmap} are, for $w\in\fr^3$:
\begin{align*}
\mybrackets{\sha{F_1}(w)}(p)&=\{\x(x,y)\leq 10,\, -\x(x,y)\leq 0,
\y(x,y)\leq 1,\, -\y(x,y)\leq -1\}^{\dagger}\\
\mybrackets{\sha{F_2}(w)}(p)&=\displaystyle{\sup_{\substack{(x,y)\in w_1^\star\\ r(x,y)\leq 0}}} p(T_{then}(x,y))\\
\mybrackets{\sha{F_3}(w)}(p)&=\displaystyle{\sup_{\substack{(x,y)\in w_1^\star\\ -r(x,y)\leq 0}}} p(T_{else}(x,y))
\end{align*}
At the first control point, for an element $w\in\fr^3$, the abstract semantics functional 
$\mybrackets{\sha{F_1}(w)}$ takes the value 10 for the template $\x$, 0 for the template $-\x$, 
1 for the template $\y$ and -1 for the template $-\y$. Again, the assignments are quadratic and 
the templates are linear, so the evaluation of the abstract semantic 
functional still reduces to a quadratic programming problem, which,
by application of Shor relaxation, leads to the following relaxed functional $\mybrackets{\rel{F_2}(w)}(\x)$ 
at control point 2, for an element $w\in\fr^3$: 
\[
\mybrackets{\rel{F_2}(w)}(\x)=\inf_{\substack{\lambda(p)\geq 0\\ \forall\, p\in P\\ \mu\in\rr_+}}\sup_{z\in\rr^2} 
\sum_{q\in P}\lambda(q)w_1(q)+z^{T}\begin{pmatrix} 0&0\\ 0&-1\end{pmatrix}
z-\sum_{q\in P}\lambda(q) q(z)+3-\mu r(z)
\]
By introducing the matrices
\[
M(\x)=\begin{pmatrix}
0 & 0.5 & 0.5\\
0.5 & 0 & 0\\
0.5 & 0 & 0
\end{pmatrix},\ 
M(\y)=\begin{pmatrix}
0 & 0.5 & -0.5\\
0.5 & 0 & 0\\
-0.5 & 0 & 0
\end{pmatrix},
\]
\[
M(-\x)=-M(\x)\text{ and } M(-\y)=-M(\y)\enspace ,
\]
\[
M(r)=\begin{pmatrix}
-2 & 0 & 0\\
0 & 1 & 0\\
0 & 0 & 1
\end{pmatrix},\qquad 
M(\x\circ T_{then})=\begin{pmatrix}
3 & 0 & 0\\
0 & 0 & 0\\
0 & 0 & -1
\end{pmatrix}
\]
and using Equation \eqref{SDP2}, we get
\[
\mybrackets{\rel{F_2}(w)}(\x)=\Min_{\substack{\lambda(p)\geq 0\\ \forall\, p\in P\\ \mu\in\rr_+ \\ \eta\in\rr}} 
\eta\, \mbox{ s.t. }\, M(\x\circ T_{then})+\eta N(-1)-\mu M(r)+\sum_{q\in P} \lambda(q) \mybrackets{N(w_{1}(q))-M(q)}
\preceq 0  
\]
We find using again Yalmip, SeDuMi and Matlab, $\mybrackets{\rel{F_2}(w)}(\x)=2$. For the lower bound of the values of "$3-y^2$",  
we find $\mybrackets{\rel{F_2}(w)}(-\x)=-1$. For the variable $y$ at control point [2], we find 
$\mybrackets{\rel{F_2}(w)}(\y)=-3.4698e-09\simeq 0$ 
and $\mybrackets{\rel{F_2}(w)}(-\y)=1$. We remark that $\mybrackets{\rel{F_2}(w)}(\y)\simeq 0$ which means that at control point
[2], the values of "$x-1$" are less or equal to 0. Indeed, Shor's relaxation scheme detects that the test is satisfied iff 
$1+x^2-2\leq 0$ which is equivalent to $x^2\leq 1$ and then the values of "$x-1$" are bounded by 0.
\end{example}




\subsubsection{Properties of the relaxed semantics}
We denote by $\affect$ the set of coordinates of the abstract semantics functional which interpret an assignment. We denote by $\inter$ the set of coordinates corresponding to tests (the symbol $\inter$ stands for ``intersection''). 
The set of coordinates $\union$ (for ``union'') corresponding to loops 
will be dealt with separately in Subsection~\ref{subsec-loop}.

Now, we are interested in the properties of the relaxed semantics which we introduced 
in the Subsection \ref{assigntestsect}. First, we start by proving that the relaxed semantics is a safe over-approximation
of the abstract semantics. It is deduced by the weak duality theorem: the relaxed semantics is a relaxation.

\begin{thm}
\label{safeapprox}
Let $i$ be a coordinate in $\affect\cup\inter$. For all $v\in\fr^n$,
\[
 \sha{F_i}(v)\leq \rel{F_i}(v)
\]
\end{thm}
The proof is given for the convenience of the reader. We only 
consider the case when $i\in\affect$ (the proof can be easily adapted to
the case of tests, i.e., $i\in\inter$).
\begin{proof}
Let $v\in\fr^n$ and $p\in P$. If $v_{i-1}^\star$ is empty (this case includes $v_{i-1}(p)=-\infty$ for some $p\in P$), 
$\mybrackets{\sha{F_i}(v)}(p)=-\infty$ for all $p\in P$ and the inequality holds. Now we suppose $v_{i-1}^\star\neq \emptyset$ and we 
take $x\in v_{i-1}^\star$. Since $\lambda\in \frp$, $\sum_{q\in P}\lambda(q)\mybrackets{v_{i-1}(q)-q(x)}\geq 0$ and then:
\[ 
p\circ T(x)\leq p\circ T(x)+\sum_{q\in P}\lambda(q)\mybrackets{v_{i-1}(q)-q(x)}\enspace .
\]
We get: 
\[
\sup\{p\circ T(x)\mid q(x)\leq v_{i-1}(q), \forall q\in P\}
\leq\sup_{x\in\rd} p\circ T(x)+\sum_{q\in P}\lambda(q)\mybrackets{v_{i-1}(q)-q(x)}\enspace .
\]
The left-hand side does not depend on $\lambda$ so we have $\mybrackets{\sha{F_i}(v)}(p)\leq \mybrackets{\rel{F_i}(v)}(p)$.
\end{proof}

We observe that the calculus of the $P$-convex hull is a special case of assignment. Indeed, to compute a $P$-convex hull,
is equivalent to compute the abstract effect of the assignment $x=x$. Consequently, we can apply the Shor's relaxation scheme 
to over-approximate the computation of the $P$-convex hull when the templates are quadratic.

\begin{cor}
\label{closurecal}
Let $w$ be in $\fr$ and $p$ in $P$ we have:
\[
\begin{array}{lcl}
\mybrackets{\clo{w}}(p) &\leq & \displaystyle{\inf_{\substack{\lambda\in\frp\\ \ \eta\in\rr}} 
\eta\, \mbox{ s.t. }\, M(p)+\eta N(-1)+\sum_{q\in P} \lambda(q) \mybrackets{N(w(q))-M(q)}
\preceq 0}\\
&\leq & w(p)\enspace .
\end{array}
\]
\end{cor}

\begin{proof}
The first inequality is a special case of Theorem~\ref{safeapprox}. 
Let us show the second inequality. Let $w\in\fr$ and $p\in P$.
We know that the following equality holds:
\[
\begin{array}{cl}
&\displaystyle{\inf_{\substack{\lambda\in\frp\\ \ \eta\in\rr}} 
\eta\, \mbox{ s.t. }\, M(p)+\eta N(-1)+\sum_{q\in P} \lambda(q) \mybrackets{N(w(q))-M(q)}
\preceq 0}\\
=&\displaystyle{\inf_{\lambda\in\frp}\sup_{x\in\rd} \sum_{q\in P}\lambda(q) w(q)
+p(x)-\sum_{q\in P}\lambda(q)q(x)}
\end{array}
\]
Let us choose $\lambda(q)=1$ if $q=p$ and 0 otherwise. Then,
$\sup_{x\in\rd} \sum_{q\in P}\lambda(q) w(q)+p(x)-\sum_{q\in P}\lambda(q)q(x)=w(p)$. Since   
we have to take the infimum over $\lambda$ so the second inequality holds.
\end{proof}


The monotonicity of the relaxed semantics will be useful for the construction of Kleene iteration. Indeed, we will 
define in the next section a simple Kleene iteration so we show that the relaxed semantics constructed for assignments 
and tests define a monotone map. We adopt the convention that $\lambda(p) w(p)=0$ if $w(p)=\infty$ and $\lambda(p)=0$ for some $p\in P$.  
\begin{prop}\label{prop-aandimonotone}
For $i\in\affect\cup\inter$, the map $v\mapsto \rel{F_i}(v)$ is monotone on the set $\fr^n$.
\end{prop}
\begin{proof}
We only give a proof for the tests so let $i$ be in $\inter$. Let $v,w$ be in $\fr^n$. 
We have for all $p\in P$, the following equality:
\[
\mybrackets{\rel{F_i}(v)}(p)=\displaystyle{\inf_{\substack{\lambda\in \frp\\ \mu\in \rr_+}}\sum_{q\in P}\lambda(q) v_{j-1}(q)
+\sup_{x\in\rd} p\circ T(x)-\sum_{q\in P}\lambda(q)q(x)-\mu r(x)}\enspace .
\]
Since $\lambda\in\frp$, the map $v\mapsto \sum_{q\in P}\lambda(q) v_{j-1}(q)$ is monotone and since 
$\sup_{x\in\rd} p\circ T(x)-\sum_{q\in P}\lambda(q)q(x)-\mu r(x)$ is a constant (independent of $v$) the function 
$v\mapsto \rel{F_i}(v)$ is an infimum of monotone maps and thus a monotone map. 
\end{proof}

Now we focus on a property which will be useful to construct a policy iteration for quadratic zones. 
We are interested in the case in which the relaxed semantics $\mybrackets{\rel{F_i}(v)}(p)$ is equal to the supremum 
in Equation~\eqref{eqn} for some $\lambda$ and equal to the supremum in Equation~\eqref{eqntest} for some pair $(\lambda,\mu)$ 
in the case of tests. A simple condition provides the desired result: if, for all $i\in\affect\cup\inter$, Slater constraint 
qualification holds i.e. there exists $x\in\rd$ such that $q(x)<v_{i-1}(q),\ \forall\, q\in P$ (and for this particular $x$, 
$r(x)<0$ holds for a test $r$), then there exists some $\lambda$ (and a couple $(\lambda,\mu)$ is the case of tests) 
which achieves the minimum in~\eqref{eqn} (the minimum in~\eqref{eqntest} in the case of tests).
Moreover the over-approximation we make is not in general that big; in some cases, Inequality \ref{safeapprox} 
is even an equality.
\begin{prop}[Selection Property]
\label{sdualth}
We assume that there exist $\lambda\in\frp$ and $\mu\in\rr_+$ such that:
\[ 
\displaystyle{\sup_{x\in\rd} p\circ T(x)+\sum_{q\in P}\lambda(q)\mybrackets{v_{i-1}(q)-q(x)}}-\mu r(x)
\] is finite.
If the set: 
\[
\{x\in\rd\mid q(x)-v_{i-1}(q)<0,\ \forall\, q\in P\}\cap \{x\in\rd\mid r(x)<0\}
\] is non-empty,
then there exist $\lambda^*\in\frp$ and $\mu^*\in\rr_+$ such that:
\[
\mybrackets{\rel{F_i}(v)}(p)=\sup_{x\in\rd} p\circ T(x)+\sum_{q\in P}\lambda^*(q)\mybrackets{v_{i-1}(q)-q(x)}-\mu^* r(x)
\]
Furthermore, if $p\circ T$ is a concave quadratic form and if for all $q\in P$ such that $v_{i-1}(q)<+\infty$, $q$ is a 
convex quadratic form, then: 
\[
\mybrackets{\rel{F_i}(v)}(p)=\mybrackets{\sha{F_i}(v)}(p)\enspace .
\]
\end{prop}
\begin{proof}
The proof is given in the appendix.
\end{proof}
In the case of assignment, we can reformulate a selection property as a particular case of Proposition~\ref{sdualth},
since we can replace $\mu$ by 0 and remove the set $\{x\in\rd\mid r(x)<0\}$. We do not give the details for that. 
\begin{remark}
We can apply Proposition~\ref{sdualth} in the case of intervals, zones and linear templates: the functions $p\in P$
are all linear in these cases. When programs contain only linear expressions in assignments and tests, 
Proposition~\ref{sdualth} implies that $F^{\sharp}=F^{\calR}$. Unfortunately, this seems to be the only simple case
in which the relaxation is exact. In other words, the linear templates of Sankaranarayanan et al. are optimal in some sense. 
Indeed assume there is only one quadratic template $p$ and consider the assignment $y=T(x)$; $T$ being affine.
To evaluate the abstract semantic, we have to solve an optimisation problem of the form:
\[
\begin{array}{cc}
\operatorname{Max}& p(y)\\
\text{s.t.}& y=T(x)\\
& p(x)\leq \alpha
\end{array}
\]
Due to the constraint $p(x)\leq \alpha$, the function $p$ must be convex for the feasible set to be convex. However the same function also 
appears as the objective function and thus, $p$ must be concave for the optimisation problem to be tractable by convex programming methods. This is possible only if $p$ is affine.
\end{remark}

The introduction of the duality provides a reformulation of Equation~\eqref{eqn} and Equation~\eqref{eqntest}:
the relaxed semantics can be rewritten as the infimum of affine functions of the function $v_{i-j}$ (when $v_{i-j}(p)\in\rr$ for all $p\in P$).
Once again, this reformulation will be useful to construct a policy iteration since it allows to solve at each step a linear 
program. So, let us fix $\lambda\in\frp$ and observe that the sum $\sum_{q\in P} \lambda(q) v_{i-1}(q)$ 
does not depend on the variable $x$ in Equation~\eqref{eqn}. 

Let $v$ be in $\fr^n$. We shall need the following notation.
\begin{iteMize}{$\bullet$}
\item For $i\in\affect$, we now define, for $\lambda\in\frp$, $F_i^{\lambda}(v)$ by: 
\begin{equation}
\label{affinepol}
\mybrackets{F_i^{\lambda}(v)}(p):=\sum_{q\in P}\lambda(q)v_{i-1}(q)+V_i^{\lambda}(p)\enspace .
\end{equation}
\begin{equation}
\label{value}
\mbox{where } V_i^{\lambda}(p):=\sup_{x\in\rd} p\circ T(x)-\sum_{q\in P}\lambda(q)q(x)
\end{equation}
\item For $i\in\inter$, we define, for $\lambda\in\frp$ and $\mu\in\rr_+$, $F_i^{\lambda,\mu}(v)$ by:
\begin{equation}
\label{affinepol2}
\mybrackets{F_i^{\lambda,\mu}(v)}(p):=\sum_{q\in P}\lambda(q)v_{i-1}(q)+V_i^{\lambda,\mu}(p) 
\end{equation}
\begin{equation}
\label{valuetest}
\mbox{where } V_i^{\lambda,\mu}(p):=\sup_{x\in\rd} p\circ T(x)-\sum_{q\in P}\lambda(q)q(x)-\mu r(x)\enspace .
\end{equation}
\end{iteMize}
The relaxed functional can now be readily rewritten as follows.
\begin{lem}
\label{lemma1}
For $i\in\affect$ and $j\in\inter$: 
\[
\mybrackets{\rel{F_i}(v)}(p)=\inf_{\lambda\in\frp}\mybrackets{F_i^{\lambda}(v)}(p),\qquad
\mybrackets{\rel{F_j}(v)}(p)=\inf_{\substack{\lambda\in\frp\\ \mu\in\rr_+}} \mybrackets{F_j^{\lambda,\mu}(v)}(p)\enspace .
\]
\end{lem}
We remark that $V_i^{\lambda}$ and $V_i^{\lambda,\mu}$ are the value of an unconstrained quadratic maximization problem. So, the functions
$V_i^{\lambda}$ and $V_i^{\lambda,\mu}$ can be determined algebraically. Moreover, $V_i^{\lambda}(p)$ and $V_i^{\lambda,\mu}(p)$ can 
take the value $+\infty$ if the matrices associated to the quadratic functions 
$x\mapsto p\circ T(x)-\sum_{q\in P}\lambda(q)q(x)$ and $x\mapsto p\circ T(x)-\sum_{q\in P}\lambda(q)q(x)-\mu r(x)$ 
are not negative semi-definite. Furthermore, the latter matrices depend on $\lambda\in\frp$ and on a couple $(\lambda,\mu)\in\frp\times\rr_+^d$.
So, to ensure the finiteness of the value, it suffices to choose $\lambda$ (or a couple $(\lambda,\mu)$ in the case 
of tests) 
such that the corresponding matrix is negative semi-definite. 
We denote by $A^{\bullet}$, the Moore-Penrose general inverse of $A$,
which can be defined as
$\lim_{\epsilon\to 0} A^T(AA^T+\epsilon Id)^{-1}$. 
The following proposition shows how to evaluate the functions $V_i^{\lambda}$
and $V_i^{\lambda,\mu}$. We only consider the evaluation of $V_i^{\lambda,\mu}$ since the evaluation of the former function can be viewed
as a special case of the evaluation of the latter.
We recall that all the assignments $T$ are such that $p\circ T$
is a quadratic function for all $p\in P$, and all 
functions $r$ arising in tests are quadratic. Moreover,
we recall that we write a quadratic function $g$ as $x\mapsto x^TA_gx+b_g^Tx+c_g$. 

\begin{prop}
\label{value2}
Let $i$ be in $\inter$ and let an assignment $T$ such that, for all $q\in P$, $q\circ T$ is a quadratic function.
Let $p$ be in $P$ and let $(\lambda,\mu)$ be a couple in $\frp\times\rr_+$, we write:
\[
\begin{array}{ccl}
\mathcal{A}_p(\lambda,\mu)&=&A_{p\circ T}-\displaystyle{\sum_{q\in P}}\lambda(q)A_q-\mu A_{r}\\
\mathcal{B}_p(\lambda,\mu)&=&b_{p\circ T}-\displaystyle{\sum_{q\in P}}\lambda(q)b_q-\mu b_{r}\\
\mathcal{C}_p(\lambda,\mu)&=&c_{p\circ T}-\displaystyle{\sum_{q\in P}}\lambda(q)c_q-\mu c_{r}
\end{array}
\]
If $(\lambda,\mu)\in\frp\times\rr_+$ satisfies $\mathcal{A(\lambda,\mu)}\preceq 0$ and $\mathcal{B}_p(\lambda,\mu)\in 
\operatorname{Im}(\mathcal{A(\lambda,\mu)})$ then:
\[
V_i^{\lambda,\mu}(p)=-\dfrac{1}{4}\mathcal{B}_p(\lambda,\mu)^T\mathcal{A}_p(\lambda,\mu)^{\bullet}\mathcal{B}_p(\lambda,\mu)
+\mathcal{C}_p(\lambda,\mu).
\]
Otherwise $V_i^{\lambda,\mu}(p)=+\infty$.
\end{prop}

The proof is classical but it is provided for the convenience of the reader. 

\begin{proof}
Let $(\lambda,\mu)\in\frp\times\rr_+$ and $p\in P$. We can rewrite 
\[V_i^{\lambda,\mu}(p)=\sup_{x\in\rd}
x^T\mathcal{A}_p(\lambda,\mu)x+\mathcal{B}_p(\lambda,\mu)^Tx+\mathcal{C}_p(\lambda,\mu).
\]
We assume that $\mathcal{A(\lambda,\mu)}$ is not a semi-definite matrix, there exists $y\in\rd$ and $\gamma\in\rr$, such that 
$\gamma^{2}y^T\mathcal{A(\lambda,\mu)}y>0$ and $\gamma\mathcal{B}_p(\lambda,\mu)^Ty\geq 0$, now taking $t>0$, we get 
$t\gamma^{2}y^T\mathcal{A(\lambda,\mu)}y>0$ and $t\gamma\mathcal{B}_p(\lambda,\mu)^Ty\geq 0$, this leads to 
$V_i^{\lambda,\mu}(p)\geq ty^T\mathcal{A}_p(\lambda,\mu)y+t\gamma\mathcal{B}_p(\lambda,\mu)^Ty+\mathcal{C}_p(\lambda,\mu)$ and 
we conclude that $V_i^{\lambda,\mu}(p)=+\infty$.

We assume that $\mathcal{A(\lambda,\mu)}\preceq 0$, the function $x\mapsto x^T\mathcal{A}_p(\lambda,\mu)x+\mathcal{B}_p(\lambda,\mu)^Tx+\mathcal{C}_p(\lambda,\mu)$
is concave and differentiable then its maximum is achieved at every zero of its derivative. A maximizer $\bar{x}$ satisfies 
$2\mathcal{A}_p(\lambda,\mu)\bar{x}+\mathcal{B}_p(\lambda,\mu)=0$. If $\mathcal{B}_p(\lambda,\mu)\notin\operatorname{Im}(\mathcal{A(\lambda,\mu)})$,
the equation $2\mathcal{A}_p(\lambda,\mu)\bar{x}+\mathcal{B}_p(\lambda,\mu)=0$ does not have a solution. Taking a non-zero 
vector $y$ in the kernel of $\mathcal{A(\lambda,\mu)}$ such that $\mathcal{B}_p(\lambda,\mu)^{T}y>0$, we get for all $\gamma>0$
that $V_i^{\lambda,\mu}(p)\geq \gamma\mathcal{B}_p(\lambda,\mu)^{T}y+\mathcal{C}_p(\lambda,\mu)$. We conclude that
$V_i^{\lambda,\mu}(p)=+\infty$.
Suppose that $\mathcal{B}_p(\lambda,\mu)\in\operatorname{Im}(\mathcal{A(\lambda,\mu)})$, then 
$\bar{x}=-\dfrac{1}{2}\mathcal{A}_p(\lambda,\mu)^{\bullet}\mathcal{B}_p(\lambda,\mu)+z$ where $z$ belongs to the kernel of
$\mathcal{A}_p(\lambda,\mu)$. Finally, we conclude that: 
\[
V_i^{\lambda,\mu}(p)=-\dfrac{1}{4}\mathcal{B}_p(\lambda,\mu)^T\mathcal{A}_p(\lambda,\mu)^{\bullet}\mathcal{B}_p(\lambda,\mu)
+\mathcal{C}_p(\lambda,\mu)
\]
since it suffices to compute $\bar{x}^T\mathcal{A}_p(\lambda,\mu)\bar{x}+\mathcal{B}_p(\lambda,\mu)^T\bar{x}
+\mathcal{C}_p(\lambda,\mu)$ to find the value of $V_i^{\lambda,\mu}(p)$.
\end{proof}

\subsection{Abstraction of loops}\label{subsec-loop}

The only point that we did not address yet is how to interpret 
the
semantics equation
at a control point $i$ in which we
collect the values of the variables before the entry in the body
of the loop, at control point $i-1$, with the values of the variables
at the end of the body of the loop, at control point $j$: 
$F_i(X)=X_{i-1}\cup X_{j}$.
By using Equation~\eqref{absmap}, for $v\in\fr^n$, $\sha{F_i}(v)=(v_{i-1}^\star\sqcup v_{j}^\star)^\dagger$.
As for zones, we notice that the union of two such $P$-convex functions $v_{i-1}$ and $v_j$ is
directly given by taking their maximum on each element of the basis of quadratic functions $P$.
Nevertheless, during the fixpoint iteration (as in Section \ref{solving}) the functions $v_{i-1}$ and $v_j$ are not 
necessarily $P$-convex. Moreover, if we take the abstract semantics $\sha{F_i}(v)$, we do not have
an infimum of linear forms (or at least a maximum of linear forms) on the abstract values $v_{i-1}$ and $v_j$, a formulation that we need. 
Finally, we relaxed the abstract semantics $\sha{F_i}(v)$, using Remark~\ref{galoisprop}, by the supremum itself
and $\rel{F_i}(v)=\sup(v_{i-1},v_{j})$. By this reduction, the map $v\mapsto \rel{F_i}(v)$ is monotone on $\fr^n$.
Recall that $\union$ denotes the set of coordinates
such that the concrete semantics is a meet operation.
For $i\in \union$, the monotonicity of the map $v\mapsto \rel{F_i}(v)$ 
follows trivially from the previous observations. 
Combining this with Proposition~\ref{prop-aandimonotone} above, we eventually get:
\begin{prop}
\label{monotonysup}
The map $v\mapsto \rel{F}(v)$ is monotone on $\fr^n$.
\end{prop}

To sum up, we conclude from Theorem~\ref{safeapprox} that we can compute over-approximations 
of~\eqref{optpb} and~\eqref{eqnabstracttest} by solving a SDP problem.

\begin{thm}
\label{polytime}
In the case of quadratic templates, for a program with
an affine arithmetics, the relaxed functional $\rel{F}$
can be evaluated using Shor's semi-definite relaxation and provides
a sound over-approximation of the abstract functional $\sha{F}$.
\end{thm}

\section{Solving the semantic equation}
\label{solving}
\subsection{Fixpoint equations in quadratic zones}

We recall that $P$ is a finite set of quadratic templates. The map $F$ is a monotone map which 
interprets a program with $d$ variables and $n$ labels in $\parvec$.
We recall that $v^{\star}$ denotes the vector of sets 
$((v_1)^{\star},\cdots,(v_n)^{\star})$ and $\sha{F}(v)=(F(v^{\star}))^{\dag}$ i.e.
$\forall\, i$, $\sha{F_i}(v)=(F_i(v^{\star}))^{\dag}$ and $\rel{F}$ is the map, the 
components of which are the relaxed functions of $\sha{F}$.
As usual in abstract interpretation, we are interested in solving the least fixpoint equation: 
\begin{equation}
\label{fixpointeq}
\inf\{v\in\vep^n\mid \sha{F}(v)\leq v\}
\end{equation}
Nevertheless, the function $\sha{F}$ is not easily computable (since the templates $p$ are polynomials, 
the epigraph of $\sha{F}$ can be checked to be a semi-algebraic set, but this of course does
not lead to scalable algorithms). Hence, we solve instead the following 
fixpoint equation in $\fr^n$:
\begin{equation}
\inf\{v\in\fr^n\mid \rel{F}(v)\leq v\}
\end{equation}
and sometimes, we will stop our analysis at some vectors $v$ such that $\rel{F}(v)\leq v$.

We next describe and compare two ways of computing (or approximating)
the smallest fixpoint of the semantic equation: 
Kleene iteration in Section \ref{Kleeneiter}, and policy iteration 
in Section \ref{politer}.
   
\subsection{Kleene iteration}
\label{Kleeneiter}

We note by $\perp$ the smallest element of $\vep^n$ i.e.
for all $i=1,\cdots,n$ and for all $p\in P$, $\perp_i(p)=-\infty$. 
The Kleene iteration sequence in $\vep^n$ is thus as follows:

\begin{enumerate}[(1)]
\item $v^0=\perp$
\item for $k\geq 0$, $v^{k+1}=\vve\circ\rel{F}(v^k)$
\end{enumerate}
Since $v\mapsto \vve\circ\rel{F}$ is monotone over $\vep^n$, 
this sequence is non-decreasing, and its limit is
a candidate to be the smallest fixpoint of the functional
$\vve\circ\rel{F}$. Unfortunately, it cannot be argued that this limit
is always the smallest fixpoint without further assumptions.
Indeed, $\rel{F}$, which is essentially defined as an infimum
of affine functions, is automatically upper semi-continuous, whereas
Scott continuity, i.e., lower semi-continuity in the present setting,
would be required to show
that the function $\rel{F}$ commutes with the supremum of increasing sequences.
However, it can be checked that the map $\vve\circ \rel{F}$ is concave
(as the composition of a concave non-decreasing function, and of a concave
function), and it is known that a concave function with values in $\br$
is continuous on any open set on which it is finite~\cite[Th.~10.1]{Roc}.
Hence, if the supremum of the sequence produced by the Kleene iteration
belongs to such an open set, it is guaranteed to be the smallest fixed
point. A more detailed theoretical analysis of the Kleene iteration,
in the present setting of non-linear templates, appears
to raise interesting technical convex analysis issues, which are beyond
the scope of this paper.


Kleene iteration has the inconvenience that the values $v^k$ which are obtained at a given iteration $k$ (before convergence) do not provide a safe invariant. We shall see that policy
iteration does not have this inconvenient: even if it is stopped at an intermediate step, it does provide a safe invariant.
Moreover, the convergence of the Kleene iteration can be very slow,
so it needs to be coupled with an 
acceleration technique which provides over-approximations.
In our implementation, after a given number of iterations,
and during a few iterations,
we round bounds outwards with a decreasing precision (akin to the widening used in~\cite{FMICS07}). 
Note also that the $P$-convex hull cannot be computed exactly,
so we over-approximate it using Shor relaxation.
This yields an approximation of the sequence $(v^k)_{k\geq 0}$, in which
the approximated vectors $v^k$ do not belong necessarily to $\vep^n$ but only to $\fr^n$. 
 
\subsection{Policy iteration algorithm}
\label{politer}
\subsubsection{Selection property and policy iteration algorithm}
A policy iteration algorithm can be used to solve a fixpoint equation for a 
monotone function written as an infimum of a family of simpler 
monotone functions, obtained by selecting {\em policies},
see~\cite{Policy1,ESOP07} for more background. The idea 
is to solve a sequence of fixpoint problems involving simpler functions.


In the present setting, we look for a representation of the relaxed
function
\begin{align}
\rel{F}= \inf_{\pi\in \Pi}F^\pi \label{e-def-select}
\end{align}
where the infimum is taken over a set $\Pi$ whose elements $\pi$ are called {\em policies}, and
where each function $F^\pi$ is required to be monotone. The correctness
of the algorithm relies on a selection property, meaning in the present
setting that for each argument $(i,v,p)$ of the function $\rel{F}$,
there must exist a
policy $\pi$ such that $\mybrackets{\rel{F}_i(v)}(p)= \mybrackets{F^\pi_i(v)}(p)$. The idea
of the algorithm is to start from a policy $\pi$, compute the smallest
fixpoint $v$ of $F^\pi$, evaluate $\rel{F}$ at point $v$,
and, if $v\neq \rel{F}(v)$, determine the new policy using the selection
property at point $v$.

Let us now identify the policies. 
Lemma~\ref{lemma1} shows that for each template $p$, each coordinate
$\rel{F_i}$ corresponding to an assignment $i\in \affect$ can be written
as the infimum of a family of affine functions $v\mapsto F_i^\lambda(v)$, 
the infimum being taken over the set of Lagrange multipliers $\lambda$. 
The same lemma provides a representation of the same nature when the coordinate
$i\in \inter$ corresponds to a test, with now a couple
of Lagrange multipliers $(\lambda,\mu)$. 
Choosing a policy $\pi$ consists
in selecting, for each $i\in \affect$ (resp. $j\in\inter$) and $p\in P$, a Lagrange multiplier $\lambda$ 
(resp.\ a pair of Lagrange multipliers $\lambda,\mu$). 
We denote by $\pi_i(p)$ (resp.\ $\pi_j(p)$) the value of $\lambda$ (resp.\ $(\lambda,\mu)$)
chosen by the policy $\pi$.

Then, the map $F^\pi$ in~\eqref{e-def-select} is obtained by replacing $\rel{F_i}$ by 
the affine functions appearing in Lemma~\ref{lemma1}, for $i\in \affect\cup\inter$. For coordinates corresponding to loops, i.e.,
$i\in \union$, we take
$F^\pi_i=\rel{F}_i$ (the choice of policy is trivial) since the infimum operation does not
appear in the expression of $\rel{F}$ (see Subsection~\ref{subsec-loop}).



%

%
Proposition~\ref{sdualth} shows that the selection property is valid
under a Slater constraint qualification condition. 
We thus introduce $\FS$, the set of elements of $\fr$ which satisfy the Slater condition
when the component $F_i$ of $F$ corresponds to an assignment or a test. More concretely:
$v\in\FS$, if, for all $i\in\affect$ the set:
\[
\{x\in\rd\mid q(x)<v_{i-1}(q),\ \forall\, q\in P\}
\]
and, for $i\in\inter$ and a test $r$, the set:
\[
\{x\in\rd\mid q(x)<v_{i-1}(q),\ \forall\, q\in P\}\cap \{x\in\rd\mid r(x)<0\}
\]
are non-empty.


\begin{algorithm}
\caption{Policy Iteration in Quadratic Templates}\label{PIQua}
\begin{iteMize}{$\bullet$}
 \item[1] Choose $\pi^0\in\Pi$,  $k=0$.
 \item[2] Compute $V^{\pi^k}=\{V^{\pi^k}(q)\}_{q\in P}$ and define
the associated function $F^{\pi^k}$ by choosing $\lambda$ and $\mu$ according
to policy $\pi^k$ in Proposition~\ref{value2} and Lemma~\ref{lemma1}.
 \item[3] Compute the smallest fixpoint $v^k$ in $\fr^n$ of $F^{\pi^k}$.
 \item[4] Compute $w^k=\clo{v^k}$. 
 \item[5] If $w^k\in\FS$ continue otherwise return $w^k$.
 \item[6] Evaluate $\rel{F}(w^k)$, if $\rel{F}(w^k)=w^k$ 
return $w^k$ otherwise take $\pi^{k+1}$ s.t. $\rel{F}(w^k)= F^{\pi^{k+1}}(w^k)$.
Increment $k$ and go to 2.  
\end{iteMize}
\end{algorithm}

This leads to Algorithm~\ref{PIQua}.
For the third step of Algorithm~\ref{PIQua}, since $P$ is finite and using Lemma ~\ref{lemma1}, 
$F^{\pi^l}$ is monotone and affine $\frf^n$, we compute the smallest fixpoint of $F^{\pi^l}$ by solving the following linear 
program see ~\cite[Section 4]{ESOP07}:
\begin{equation}
\label{LPfix}
\min \sum_{i=1}^n\sum_{q\in P}v^i(q)\ \mathrm{s.t.}\ \mybrackets{F_k^{\pi_k^l}(v)}(q)\leq v_k(q),\ \forall k=1,\cdots,n,
\ \forall q\in P
\end{equation}

\begin{remark} 
As in the case of the earlier policy iteration algorithms in static analysis~\cite{Policy1,ESOP07}, an important
issue is the choice of the initial policy, which may influence the quality
of the invariant which is eventually determined.
In~\cite{Policy1,ESOP07}, the initial policy was selected by assuming
that the infimum is the expression of the functional is attained by
terms corresponding to guard conditions, see specially \S~4.2
in~\cite{ESOP07}. The same principle can be used here. 
Another method to choose an initial policy is to run a few
Kleene iterations, in combination with an acceleration
technique. This leads to a postfixpoint $v$ of $\rel{F}$,
and we select as the initial policy any policy attaining the infimum
when evaluating $\rel{F}(v)$ (i.e., choose for $\pi_i(p)$ or $\pi_j(p)$
any Lagrange multiplier $\lambda$ or pair of Lagrange multipliers
$\lambda,\mu$ attaining the infimum in Lemma~\ref{lemma1}).
%
\end{remark}

\begin{remark}
To ensure the feasibility of the solution of ~\eqref{LPfix} computed by the LP solver, we replace, when possible, 
the constraint set by $F^{\pi^l}(v)+\epsilon\leq v$, where $\epsilon$ is a small constant 
(typically of the order of several $ulp(v)$, where $ulp(v)$, which stands
for ``unit of least precision'', is the minimum over the coordinates $i$ of the differences 
between the nearest floating points around $v_i$).

To obtain safe bounds even though we run our algorithm on machines which uses finite-precision arithmetic, we should
use a guaranteed LP solver (e.g. LURUPA see ~\cite{lurupa}) to check that the solution obtained verifies $F^{\pi^l}(v)\leq v$.
\end{remark}

In the fourth step of Algorithm ~\ref{PIQua}, the operation of closure is, in practice, the relaxation of the $P$-convex 
hull computed by a SDP solver (see Corollary ~\ref{closurecal}). The same corollary shows that the SDP relaxation of the 
$P$-convex hull of some $w\in\fr$, is still smaller than $w$ and this result ensures a gain of precision. 

We can only prove that policy iteration on quadratic templates converges (maybe in infinite time) towards 
a postfixpoint of our abstract functional and that under some technical conditions, it converges 
towards a fixpoint. One interest in policy iteration for static analysis is that we can always terminate 
the iteration after a finite time, and end up with a postfixpoint. 

\begin{thm}
\label{decro}
The following assertions hold:
\begin{enumerate}[\em(1)]
\item $\rel{F}(v^l)\neq v^l\implies \rel{F}(v^l)< v^l$;
\item The sequence $v^l$ computed by Algorithm~\ref{PIQua} is strictly decreasing;
\item The limit $v^\infty$ of the sequence $v^l$ is a postfixpoint: $\rel{F}(v^\infty)\leq v^\infty$.
\end{enumerate}
\end{thm}

\begin{proof}
(1).~Let $l\in\mathbb{N}$. We assume that $l>0$ and $\rel{F}(v^l)\neq v^l$, there exists $\pi^l$ such that, $F^{\pi^l}(v^l)=v^l$ and since 
$\rel{F}=\inf F^{\pi}$, we get $\rel{F}(v^l)\leq F^{\pi^l}(v^l)=v^l$ and from $\rel{F}(v^l)\neq v^l$, we conclude that $\rel{F}(v^l)< v^l$.

(2).~We prove the second point by induction on $l\in\mathbb{N}$. We suppose that $l=0$ if $\rel{F}(v^0)\neq v^0$ and that 
$v^0\in\FS$ otherwise the algorithm stops. There exists $\pi^{1}$ such that $\rel{F}(v^0)=F^{\pi^{1}}(v^0)<v^0$ by the point 
1. Moreover, $v^{1}$ is the smallest element of $\{v\in\fr^n\mid F^{\pi^{1}}(v)\leq v\}$ thus $v^{1}\leq v^{0}$ and since
$F^{\pi^{1}}(v^0)<v^0$ we conclude that $v^{1}<v^{0}$. The same argument holds if $v^{l}<v^{l-1}$ and $v^l\in\FS$. 

(3).~Finally, we deduce from $v^\infty\leq v^l$ that $\rel{F}(v^\infty)\leq \rel{F}(v^l)\leq v^l$. Taking the infimum over $l$, we get $\rel{F}(v^\infty)\leq v^\infty$. 
\end{proof}
\begin{remark}
It is desirable to choose (if possible) the initial policy $\pi^0$ so that the set: 
\[
\{v\in\frf^n\mid \mybrackets {F_i^{\pi_i^0}(v)}(q)\leq v_i(q),\ \forall i=1,\cdots,n,\ \forall q\in P\} 
\]
is non-empty. Indeed, the non-emptyness of this set ensures that the coordinates of the first vector $v^0$ are not 
equal to $+\infty$ and then, by Theorem \ref{decro}, all the terms $w^k$ of the sequence generated by the policy iteration
have coordinates which are not equal to $+\infty$. Then, the policy iteration
algorithm at any step $k$ and at any breakpoint $i$ returns non-trivial invariants of the form $q(x) \leq \alpha$ (with $\alpha:=w^k_i(q)$ finite).
\end{remark}

\begin{remark}
The policy iteration algorithm developed in~\cite{Policy1,ESOP07}
can be recovered as a special case of Algorithm~\ref{PIQua},
when applied to a domain of linear templates
or to the domains of zones or intervals, 
for a program containing only linear expressions in assignments and tests.
Indeed, the main addition in the present algorithm is the presence
of the relaxation, and the latter turns out to be exact in these special cases,
see Proposition~\ref{sdualth}.

%
\end{remark}

\subsection{Max strategy iteration}
Gawlitza and Seidl \cite{DBLP:conf/sas/GawlitzaS10} developed an alternative iteration to compute the
least fixpoint of the relaxed semantics for the quadratic templates. The relaxed semantics which they 
use are constructed from the dual program of Shor's relaxation SDP problem \eqref{shordual}. Our relaxed 
semantics coincide with their relaxed semantics when technical conditions (which are often satisfied)
hold. They obtain a map whose coordinates are the maximum of a finite number of concave functions. 
Their approach consists in solving the least fixpoint equation from below and they initialize
their iteration by the function which is identically equal to $-\infty$. At each step
of their algorithm, they select a function which achieves the maximum. Since they compute the 
least fixpoint from below, the least fixpoint is returned. This implies that their iteration
must be run until a fixpoint is reached whereas our approach allows to stop the iteration at each 
step of the algorithm to provide a valid invariant. A survey~\cite{ConvOpt} recapitulates 
the two approaches.
\subsection{A detailed calculation on the running example}

Now we give details on the harmonic oscillator of Example~\ref{running}.
The program of this example which is given at Figure \ref{Exfinal2} implements an Euler explicit scheme with 
a small step $h=0.01$, that is, which simulates the linear system $(x,v)^T= T(x,v)^T$ with 
\[
T=
\begin{pmatrix}
1&h\\
-h&\ 1-h 
\end{pmatrix}
\]
We want to use the information of a Lyapunov function $\lya$ of the linear system $T$ to compute bounds 
on the values taken by the variables $x$ and $v$ of the simulation:
the function $(x,v):\mapsto (x,v)L(x,v)^{T}$ furnishes a Lyapunov function with 
\[
L= 
\begin{pmatrix}
2&1\\
1&3 
\end{pmatrix}
\] 
We also use the quadratic functions $(x,v)\mapsto x^2$ and $(x,v)\mapsto v^2$ which corresponds 
to interval constraints. We introduce the set of templates $P=\{\x,\vv,\lya\}$ and below
the program it is described the semantic equations for all the three control points.
\begin{figure}
\begin{lstlisting}[frame=single]
x = [0,1];
v = [0,1]; [1] 
h = 0.01; 
while (true) { [2] 
  u = v;
  v = v*(1-h)-h*x;
  x = x+h*u; [3] }
\end{lstlisting}
\[
\begin{array}{rcl}
\sha{F_1}(w)(p)&=&\{\x(x,v)\leq 1,\, \vv(x,v)\leq 1,\, \lya(x,v)\leq 7\}\\
\sha{F_2}(w)(p)&=& (\sup\{w_1^\star(p),w_3^\star(p)\})^{\dagger}\\
\sha{F_3}(w)(p)&=&\displaystyle{\sup_{(x,v)\in(w_2)^{\star}}} p(T(x,v))
\end{array}
\]
\caption{Implementation of the harmonic oscillator and its semantics in $\fr^3$}
\label{Exfinal2}
\end{figure}
Now we are going to focus on the third coordinate of $\mybrackets{\rel{F}(v)}(p)$.
Let us consider, for example, $p=\x$, we get:
$\mybrackets{\rel{F_3}(v)}(\x)=$
\begin{equation}
\label{example}
\inf_{\lambda\in\frp}\sup_{(x,v)\in\rr^2}\sum_{q\in P} \lambda(q)w_2(q)
+(x,v)\left(\begin{pmatrix} 1-\lambda(\x)&h/2\\ h/2&h^2-\lambda(\vv)\end{pmatrix}-\lambda(\lya)L\right)(x,v)^T
\end{equation}
By introducing the following symmetric matrices, we can rewrite ~\eqref{example} as Equation ~\eqref{runningSDP}:
\[
M(\x)=\begin{pmatrix}
0 & 0 & 0\\
0 & 1 & 0\\
0 & 0 & 0
\end{pmatrix},\ 
M(\vv)=\begin{pmatrix}
0 & 0 & 0\\
0 & 0 & 0\\
0 & 0 & 1
\end{pmatrix} \text{ and }
M(\x\circ T)=\begin{pmatrix}
0 & 0 & 0\\
0 & 1 & h/2\\
0 & h/2 & h^2
\end{pmatrix}
\]
\begin{equation}
\label{runningSDP}
\mybrackets{\rel{F_3}(w)}(\x)=\Min_{\substack{\lambda\in\frp\\ \ \eta\in\rr}} 
\eta\, \mbox{ s.t. }\, M(\x\circ T)+\eta N(-1)+\sum_{q=\x,\vv,\lya} \lambda(q) \mybrackets{N(w_{2}(q))-M(q)}
\preceq 0  
\end{equation}
To initialize Algorithm~\ref{PIQua}, we choose a policy $\pi^0$. For the third coordinate
of $\rel{F}$, we have to choose a policy $\pi_3^0$ such that $V_3^{\pi_3^0}(p)$ is finite for every $p=\x,\vv,\lya$. 
We can start, for example, by:
\[
\pi_3^0(\x)=(0,0,1),\ \pi_3^0(\vv)=(0,0,1),\ \pi_3^0(\lya)=(0,0,1)\enspace .
\]
This consists, for $p=\x$, in taking $\lambda(\x)=\lambda(\vv)=0$ and $\lambda(\lya)=1$ in ~\eqref{example}.
By Proposition \ref{value2} we find:
\[
\begin{array}{c}
V_3^{\pi_3^0}(\x )=\sup_{(x,v)\in\rr^2}(x,v)\begin{pmatrix} -1&h/2-1\\ h/2-1&h^2-3\end{pmatrix}(x,v)^T=0\\ 
\\
V_3^{\pi_3^0}(\vv)=\sup_{(x,v)\in\rr^2}(x,v)\begin{pmatrix} h^2-2&h(1-h)-1\\ h(1-h)-1&(1-h)^2-3\end{pmatrix}(x,v)^T=0\\
\\
V_3^{\pi_3^0}(\lya)=\sup_{(x,v)\in\rr^2}(x,v) (T^TLT-L) (x,v)^T=0
\end{array}
\]
The solution of the maximization problems are zero since all the three matrices are negative definite (i.e.
a matrix $B$ is negative definite iff $x^t A x<0$ for all $x\neq 0$). The third matrix $T^TLT-L$ is negative definite 
since $L$ satisfy the Lyapunov condition for the discrete linear system $(x,v)=T(x,v)$.
To compute the least fixpoint of $F^{\pi^0}$, we solve the following linear program (see~\eqref{LPfix}): 
\[\underset{1\leq\beta_1(\x),\ 1\leq\beta_1(\vv),\ 7\leq\beta_1(\lya)}{
\underset{
1\leq\beta_2(\x),\ 1\leq\beta_2(\vv),\ 7\leq\beta_2(\lya)}{
\underset{
\beta_3(\x)\leq\beta_2(\x),\ \beta_3(\vv)\leq\beta_2(\vv),\ \beta_3(\lya)\leq\beta_2(\lya)}{
\underset{
\beta_2(\lya)\leq \beta_3(\x),\ \beta_2(\lya)\leq \beta_3(\vv),\ \beta_2(\lya)\leq \beta_3(\lya)}
{\operatorname{\min} \sum_{i=1}^3\sum_{p\in P}\beta_i(p)}}}}\]
Using solver \texttt{Linprog}, we find:
\[
\begin{array}{llcllcll}
u_1^0(\x)=&1.0000 & \mbox{ } \mbox{ } \mbox{ } & u_2^0(\x)=&7.0000 & \mbox{ } \mbox{ } \mbox{ } &u_3^0(\x)=&7.0000\\
u_1^0(\vv)=&1.0000& \mbox{ } \mbox{ } \mbox{ }& u_2^0(\vv)=&7.0000 & \mbox{ } \mbox{ } \mbox{ } &u_3^0(\vv)=&7.0000\\
u_1^0(\lya)=&7.0000& \mbox{ } \mbox{ } \mbox{ }& u_2^0(\lya)=&7.0000& \mbox{ } \mbox{ } \mbox{ } & u_3^0(\lya)=&7.0000\\
\end{array}
\]
The approximation of the closure of $u^0$, $\clo{u^0}$ is given by a \texttt{Matlab} implementation, using 
\texttt{Yalmip} and \texttt{SeDuMi} returns the vector $w_1^0$:
\[
\begin{array}{llcllcll}
w_1^0(\x)=&1.0000 & \mbox{ } \mbox{ } \mbox{ } & w_2^0(\x)=&4.2000 & \mbox{ } \mbox{ } \mbox{ } &w_3^0(\x)=&4.2000\\
w_1^0(\vv)=&1.0000& \mbox{ } \mbox{ } \mbox{ }& w_2^0(\vv)=&2.8000 & \mbox{ } \mbox{ } \mbox{ } &w_3^0(\vv)=&2.8000\\
w_1^0(\lya)=&7.0000& \mbox{ } \mbox{ } \mbox{ }& w_2^0(\lya)=&7.0000& \mbox{ } \mbox{ } \mbox{ } & w_3^0(\lya)=&7.0000\\
\end{array}
\]
Using again \texttt{Yalmip} with the solver \texttt{SeDuMi}, the vector $w$ is not a fixpoint of $\rel{F}$, so we get the 
new following policy:
\begin{equation*}
\begin{aligned}
\pi_3^1(\x)=(0,0,0.596),\ 
\pi_3^1(\vv)=(0,0,0.3961),\ 
\pi_3^1(\lya)=(0,0,0.9946)
\end{aligned}\enspace .
\end{equation*}
Finally, after 5 iterations we find that the invariant of the loop i.e. $w_2^{\star}$ at control point 2 is the set:
\[
\{x^2\leq 3.5000,\ v^2\leq 2.3333,\ 2x^2+3v^2+2xv\leq 7\}\enspace.
\] 
We draw $w_2^{\star}$ at each iteration of Algorithm~\ref{PIQua} in Figure \ref{appB}.

\begin{center}
\begin{figure}
\begin{center}
\input backgroundtex
\caption{Successive templates along policy iteration, at control point 2, for the harmonic oscillator.}
\label{appB}
\end{center}
\end{figure}
\end{center}

This method is to be compared with the classical Kleene iteration with widening. 
On this example, we find without widening $x^2\leq 3.5000$, $v^2 \leq 2.3333$ and $2x^2+3v^2+2xv\leq 7$ in 1188 iterations
whereas with the acceleration technique described Subsection \ref{Kleeneiter} we find 
$x^2\leq 6.0000$, $v^2 \leq 4.0000$ and $2x^2+3v^2+2xv\leq 10$ in 15 iterations.

\section{Benchmarks}
We implemented an analyzer for the quadratic template domain we presented, 
written in Matlab version 7.8(R2009a).
This analyzer takes a text file in argument. This text file corresponds
to the abstract equation $v=\sha{F}(v)$ where $\sha{F}$ is defined by Equation~\eqref{abstraction}.
The quadratic template can be loaded from a \texttt{dat} file by the analyzer. 
The affine maps are treated in the same manner.  

In this analyzer, we can choose to use the Kleene iteration method
or policy iteration. For the Kleene iteration method, the user gives as an argument a maximal number of iteration and if
the acceleration method has to be applied. The acceleration method start from the iteration $n+1$ ($n$ denotes
the number of lines of the code) and ends eleven iteration after. For the policy iteration method, the user gives the 
\texttt{dat} file defining the initial policy or chooses to make Kleene iterations before determining the initial policy.

For the policy iteration, the user gives also as argument a maximal number of iteration and the policy iteration runs 
until a fixpoint is reached or Slater constraint qualification is no longer satisfied or if the maximal number
of iteration is achieved and so a postfixpoint can be returned by the policy iteration algorithm.
Similarly, the Kleene iteration with acceleration provides a postfixpoint after acceleration and widening to top, 
if the iteration does not converge after a given number of iterations.
The analyzer writes, in a text file, information about time, quality of the invariants found and number of iterations.

For the benchmarks, we used a single core of a laptop PC Intel(R) Duo CPU
P8600 at 2.4 Ghz with a main memory of 4Gb. We indicate in Table~\ref{tablebench}, the name of the program analyzed, the method
used (policy iteration or Kleene iteration) for solving the fixpoint equation, the cardinality of the basis of 
quadratic templates used, the number of lines of C code the program has, the number of variables it manipulates, the
number of loops. Then we indicate the number of iterations made, whether it reaches a fixpoint or (strictly) a postfixpoint.
Finally, the last column concerns the time in seconds to compute the invariant. 
\begin{center}
\begin{figure}
\begin{center}
\begin{tabular}{|ccc|ccc|c|c|c|c|}
\hline
Programs & Method &\#P& \#lines& \#var& \#loops& \#Iter.& Inv. quality & Time (s)\\
\hline
Rotation2&Policy&2&2&2&0&0&Fixpoint& 2.85\\
\hline
Rotation2& Kleene&2&2&2&0&2&Fixpoint& 2.87\\
\hline
Rotation10&Policy&2& 2&10&0&0&Fixpoint& 2.85\\
\hline
Rotation10&Kleene&2&2&10&0&2&Fixpoint&3.04\\
\hline
Filter&Policy&5&3&2&1&3&Fixpoint& 9.93\\
\hline
Filter&Kleene&5&3&2&1&15&Postfixpoint&64.03\\
\hline
Oscillator&Policy&3&3&2&1&5&Fixpoint& 9.92\\
\hline
Oscillator&Kleene&3&3&2&1&15&Fixpoint&38.91\\
\hline
Oscillatorc2&Policy&3&3&4&1&5&Fixpoint& 10.09\\
\hline
Oscillatorc2&Kleene&3&3&4&1&15&Postfixpoint&40.54\\
\hline
Oscillatorc5&Policy&3&3&10&1&5&Fixpoint& 11.43\\
\hline
Oscillatorc5&Kleene&3&3&10&1&16&Fixpoint&57.33\\
\hline
Oscillatorc10&Policy&3&3&20&1&6&Fixpoint& 22.03\\
\hline
Oscillatorc10&Kleene&3&3&20&1&20&Postfixpoint&161.43\\
\hline
Oscillatorc20&Policy&3&3&40&1&6&Fixpoint& 236.40\\
\hline
Oscillatorc20&Kleene&3&3&40&1&20&Fixpoint&1556.90\\
\hline
Symplectic&Policy&5&3&2&1&0&Fixpoint& 4.22\\
\hline
Symplectic&Kleene&5&3&2&1&15&Fixpoint& 66.14\\
\hline
SymplecticSeu&Policy&5&3&2&1&5&Postfixpoint& 12.79\\
\hline
SymplecticSeu&Kleene&5&3&2&1&15&Postfixpoint& 66.02\\
\hline
\end{tabular}
\caption{Benchmarks}
\label{tablebench}
\end{center}
\end{figure}
\end{center}

The file Rotation10 is the problem of Example~\ref{rotation} in dimension 10.
By the fixpoint computation, we prove automatically that the unit sphere in
dimension 10 is invariant by rotation. Both Kleene iteration and policy iteration find the unit sphere as invariant.

The program Filter is an implementation of recursive linear filter of second order. The program is described
Figure ~\ref{filterprog}.
\begin{figure}
\begin{lstlisting}[frame=single]
x = [0,1];
y = [0,1]; [1]
while [2] (true) { 
  x = (3/4)*x-(1/8)*y;
  y = x; [3]
}
\end{lstlisting}
\caption{The program Filter}
\label{filterprog}
\end{figure}
By policy iteration, we find the following set as loop invariant (at control point 2):
\[
\{-0.5\leq x\leq 1,\ -0.5\leq v\leq 1,\ 3x^2+v^2\leq 4\}\enspace .
\]
Although, the Kleene iteration with acceleration finds the following set:
\[
\{-1.8257\leq x\leq 1.8257,\ -3.1623\leq v\leq 3.1623,\ 3x^2+v^2\leq 10\}\enspace .
\]
The program Oscillator is the problem~\ref{running}. The invariant depicted Figure~\ref{running} in Section \ref{intro} 
is found by policy iteration whereas Kleene iteration after applying acceleration techniques 
from the iteration 4 to iteration 15 finds the less precise invariant 
$\{x^2\leq 6.0000,\ v^2\leq 4,\ 2x^2+3v^2+2xv\leq 10\}$, in 
more time.

In order to illustrate the scalability of the method, we considered a higher dimensional analogue of the example 
of Figure~\ref{running}, modelling $N$ coupled harmonic oscillators, $\ddot{x}_i+\dot{x}_i+x_i+\epsilon \sum_{1\leq j\leq N} x_j=0$,
for $1\leq i\leq N$. Hence, $N$ variables $x_i$, $v_i$, and $w_i$ now appear in the discretised scheme instead of $x,v,w$.
We computed automatically a Lyapunov function in \texttt{Matlab}. We used as templates the one arising from the latter 
Lyapunov function, together with $\sum_{i=1}^N x_i^2$ and $\sum_{i=1}^N v_i^2$. We took $\epsilon=0.5$. We made 
tests successively for $N=2,5,10,20$ (Oscillatorc2, Oscillatorc5, Oscillatorc10 and Oscillatorc20). 
As in the program Oscillator, we were interested in the loop invariant. For example, for $N=20$, 
we found with policy iteration algorithm the following set:
\[
\{(x,v)\in\rr^{20}\times\rr^{20}\mid \sum_{i=1}^{20} x_i^2\leq 115.7169,\ \sum_{i=1}^{20} v_i^2\leq 266.3048,\ 
(x,v)L(x,v)^T\leq 350.0690\}
\]  
where $L$ denotes positive semi-definite matrix associated to the Lyapunov function found automatically 
by \texttt{Matlab}.

The Symplectic example implements a discretisation of $\ddot{x}+c\dot{x}+x=0$ with $c=0$
by a symplectic method, considering specially the
case in which $c=0$ (there is no damping).
Then, the dynamical system has imaginary eigenvalues (its orbits are circles),
reflecting the fact that energy is constant.  However,
the Euler scheme, which does not preserve this conservation law,
diverges, so we use a symplectic discretisation scheme (preserving the symplectic form, 
see~\cite{Hairer03}). This is an interesting, highly degenerate,
numerical example from the point of view of 
static analysis, because there is no ``stability margin'',
hence, methods not exploiting the Lyapunov function are likely to produce trivial invariants when $c=0$.
As in Oscillator, we start from a position $x\in [0,1]$ and a speed $v\in [0,1]$.
The discretisation of $\ddot{x}+x=0$ with the symplectic method and a step $\tau=0.1$
gives us the matrix $T$ such that $T_{1,1}=1-\frac{\tau}{2}$, $T_{1,2}=\tau-\frac{\tau^3}{4}$,
$T_{2,1}=-\tau$ and $T_{2,2}=1-\frac{\tau}{2}$.
We use the Lyapunov function $L$ such that $L(x,v)=(x,v)Q(x,v)^T$
with \[Q=\begin{pmatrix}
1& & &0\\
0& & &1-\frac{\tau^2}{4}
\end{pmatrix}\enspace.
\]
The symplectic method ensures that $L(T(x,v))=L(x,v)$. Our method takes
advantage of this conservation law, since $L$ is embedded as a template.
\begin{figure}
\begin{lstlisting}[frame=single]
tau = 0.1;
x = [0,1];
v = [0,1]; [1]
while [2] (true) { 
  x = (1-(tau/2))*x+(tau-((tau^3)/4))*v;
  v = -tau*x+(1-(tau/2))*v; [3]
}
\end{lstlisting}
\caption{An implementation of the symplectic method} 
\end{figure}
The policy iteration returns the following as invariant set at the control point 2: 
\[
\{-1.41333\leq x\leq 1.41333,\ -1.4151\leq v\leq 1.4151,\ x^2+0.9975v^2\leq 1.9975\}\enspace .
\]
Whereas the Kleene iteration returns:
\[
\{-3.16624\leq x\leq 3.16624,\ -3.16628\leq v\leq 3.16628,\ x^2+0.9975v^2\leq 10\}\enspace .
\]
which is less precise. In particular, the Kleene algorithm misses
the invariance of the Lyapunov function.
 
SymplecticSeu is a symplectic method with a threshold on $v=\dot{x}$.
We iterate the Symplectic method while $v\geq\frac{1}{2}$, which gives 
the following code:
\begin{lstlisting}[frame=single]
x = [0,1];
v = [0,1];
tau = 0.1 [1]
while [2] ((v>=1/2)) {
  x = (1-tau/2)*x+(tau-(tau^3)/4)*v;
  v = -tau*x+(1-tau/2)*v; [3]
};
\end{lstlisting}
The policy iteration returns the following set which the invariant found at control point 2: 
\[
\{0\leq x\leq 1.3654,\ 0\leq v\leq 1,\ x^2+0.9975v^2\leq 1.9975\}\enspace .
\]
and the policy iteration returns the following set at control point 3:
\[
\{0.0499\leq x\leq  1.3222,\ 0.5\leq v\leq 0.9950,\ x^2+0.9975v^2\leq 1.9975\}\enspace .
\]
Although, the Kleene iteration with acceleration provides the following set which the invariant found 
at control point 2: 
\[
\{0\leq x\leq 3.1623,\ 0\leq v\leq 3.1662 ,\ x^2+0.9975v^2\leq 10\}\enspace .
\]
and the same set at the control point 3.




\begin{remark}
In the present benchmarks, the execution time of a single policy iteration step
and of a single Kleene iteration step are comparable 
(the acceleration provided by policy iteration comes from the smaller
number of iterations). Indeed, at each Kleene or Policy iteration,
the bottleneck appears to be the evaluation of the
relaxed functional $\rel{F}$, which requires to solve a family
a Shor SDP relaxations. Policy iteration requires in addition to solve
a family of linear programs (to compute the smallest fixpoint
of the current policy). This is technically easier than solving the SDPs,
and generally faster. However, each SDPs is typically local (involving a small number of variables), whereas the linear programs, which couple all the breakpoints and all the templates functions, may be of large size. Hence, one may construct (very large) instances in which solving the linear programs would become the bottleneck.
\end{remark}

\section{Conclusion and future work}

We have presented in this paper a generalization of the linear templates
of Sankaranarayanan et al.~\cite{Sriram1,Sriram2} allowing one to
deal with non-linear templates. 
We showed that in the case of quadratic templates, we can
efficiently abstract the semantic functionals using Shor's relaxation, and
compute the resulting fixpoint using policy iteration. 
Future work include the use of tighter relaxations for quadratic problems.
In particular, sum of squares (SOS) relaxations (see for instance \cite{lasserre} and~\cite{parillo}) would allow us to obtain more accurate safe over-approximation of the abstract
semantic functional for arbitrary polynomial templates and 
for a general program arithmetics. Kleene iteration could be easily
implemented in that way. However, the issue of coupling SOS relaxations
with policy iteration appears to be more difficult. Note also
that unlike Shor relaxation, SOS relaxations are subject 
to a ``curse of dimensionality''. These issues will be examined
further elsewhere.
Another problem is to extend the minimality result of \cite{Policy2} which
is currently only available for the interval domain, to our template domain.  
We intend to study more in-depth the complexity issues raised by our general policy iteration algorithm.  
Finally, we note that a more detailed account of the present
work has appeared in the Phd thesis of the first author~\cite{adjephd}.

\section*{Acknowledgement.}
We thank Thomas Gawlitza and David Monniaux for their
remarks on an earlier version of this paper. We also thank the three
referees for their careful reading and comments.
\bibliographystyle{alpha}
\bibliography{cavbib}

\newpage

\section*{Appendix}
Let $f$ be a function from $\rd$ to $\rr\cup\{+\infty\}$. We say that the function $f$ is proper 
iff there exists $x\in\rd$ such that $f(x)\in\rr$.

We recall that $f$ is convex iff the set $\operatorname{epi}(f):=\{(x,\alpha)\in\rd\times\rr\mid f(x)\leq \alpha\}$ 
is a convex set (i.e. for all $t\in [0,1]$, for all $x,y\in \operatorname{epi}(f)$, $tx+(1-t)y\in \operatorname{epi}(f)$).
Since the intersection of a family of convex sets is a convex set then the pointwise supremum of
a family of convex functions is also a convex function. 
 
The function $f$ is said to be lower semi-continuous iff the sets $\{x\in\rd\mid f(x)\leq \alpha\}$
are (topologically) closed for all $\alpha\in\rr$. Since the intersection of a family of closed sets 
is closed set then the pointwise supremum of a family of lower semi-continuous functions is also a 
lower semi-continuous function. 

The function $f$ is level-bounded iff $\{x\in\rd\mid f(x)\leq \alpha\}$ are bounded set for all $\alpha\in\rr$. 
If $\displaystyle{\lim_{\norm{x}\to +\infty}} f(x)=+\infty$ then $f$ is level-bounded. 

Proposition \ref{sdualth} is deduced from a well-known result of convex optimization:

\begin{prop}
\label{optim_result}
Let $f:\rd\mapsto \rr\cup\{+\infty\}$ be a convex, proper, lower semi-continuous 
and level bounded function. Then $\displaystyle{\inf_{x\in\rd}} f(x)$ is finite and there exists 
$\bar{x}$ such that:
\[
f(\bar{x})=\inf_{x\in\rd} f(x)\enspace .
\]
\end{prop} 
A proof of Proposition \ref{optim_result} can be found in \cite[Proposition 3.1.3]{AuTe}.

\begin{proof}[Proof of Proposition \ref{sdualth}]

We begin the proof by writing:
\[
g(\lambda,\mu)=
\left\{
\begin{array}{lr}
\displaystyle{\sup_{x\in\rd} p\circ T(x)+\sum_{q\in P}\lambda(q)\mybrackets{v_{i-1}(q)-q(x)}}-\mu r(x)& 
\text{if }\lambda\in\frp,\mu\in\rr_+\\
+\infty& \text{ otherwise}
\end{array}
\right.
\]
To show Proposition \ref{sdualth}, it suffices to show that the hypothesis of Proposition
\ref{optim_result} are verified: the function $g$ is convex, lower semi-continuous, proper and level-bounded.
 
The set $P$ is finite so $\fr$ is the finite dimensional vector space $\rr^{|P|}$ where $|P|$ is the cardinality of $P$.
First, the function $g$ is convex and lower-continuous as the pointwise supremum of convex continuous functions.
 
The function $g$ is also proper since $g(\lambda,\mu)>-\infty$, for all $\lambda\in\frp$, for all 
$\mu\in\rr_+$ and there exist $\lambda\in\frp$ and $\mu\in\rr_+$ such that $g(\lambda,\mu)<+\infty$.
 
The function $g$ is level-bounded since $\lambda\in\frp$, $\mu\in\rr_+$ and there exists 
some $\bar{x}\in\rd$ such that, for all $q\in P$, $v_{i-1}(q)-q(\bar{x})>0$ and $r(\bar{x})<0$, we conclude that 
$g(\lambda,\mu)\geq p\circ T(\bar{x})+\min_{q\in P}(v_{i-1}(q)-q(\bar{x}))\sum_{q\in P}\lambda(q)
-r(\bar{x})\mu$ then $g(\lambda,\mu)$ tends to $+\infty$ as 
$\max(\norm{\lambda}_1,|\mu|)$ tends to $\infty$. 

The second part of Proposition~\ref{sdualth} follows from the strong duality theorem for convex 
optimization problems, see e.g. ~\cite[Proposition 5.3.2]{AuTe}. 
\end{proof}


\end{document}